\documentclass[aps,pra,twocolumn,10pt]{revtex4-2}
\pdfoutput=1
\usepackage[T1]{fontenc}
\usepackage{amsthm}
\usepackage{amsmath}
\usepackage{latexsym}
\usepackage{amsfonts}
\usepackage{amssymb}
\usepackage[mathscr]{eucal} 
\usepackage{color}
\usepackage{bm,bbm,dsfont}
\usepackage{thmtools}
\declaretheorem[name=Theorem]{theorem}
\usepackage{graphicx}
\usepackage[colorlinks]{hyperref}
\usepackage{etoolbox}
\usepackage[nameinlink,capitalise]{cleveref} 
\definecolor{su-green}{RGB}{16, 126, 59}
\definecolor{su-blue}{RGB}{0,96,170}
\definecolor{su-orange}{RGB}{255,140,0}
\usepackage{tikz}
\usepackage{array} 
\newcolumntype{P}[1]{>{\centering\arraybackslash}p{#1}}
\usetikzlibrary{shapes.geometric, arrows,calc,through,positioning,matrix,graphs,graphs.standard,arrows.meta}

\usepackage[caption=false]{subfig} 
\usepackage{booktabs,tabularx}
\usepackage{longtable}

\usepackage{placeins}

\newtheorem{proposition}{Proposition}
\newtheorem{proposition?}{Proposition?}
\newtheorem{lemma}{Lemma}
\newtheorem{corollary}{Corollary}

\newtheorem{conjecture}{Conjecture}
\newtheorem*{restate}{Theorem}
\newtheorem*{proposition*}{Proposition}
\newtheorem*{corollary*}{Corollary}

\theoremstyle{definition}

\newtheorem{remark}{Remark}

\newtheorem{definition}{Definition}

\newcommand{\no}[1]{\left\|#1\right\|} 
\newcommand{\tr}[1]{\textrm{tr}\left[#1\right]} 
\newcommand{\rank}{\mathrm{rank}\,} 

\newcommand{\id}{\mathbbm{1}} 

\newcommand{\E}{\mathsf{E}}
\renewcommand{\P}{\mathsf{P}}
\newcommand{\R}{\mathsf{R}}
\newcommand{\M}{\mathsf{M}}

\usepackage{tikz}
\definecolor{darknavy}{RGB}{0,34,102} 

\tikzset{
  vtx/.style={circle,fill=darknavy,inner sep=0pt,minimum size=3pt},
  edg/.style={line width=0.4pt, darknavy}
}

\usepackage{tikz}
\usetikzlibrary{calc}

\tikzset{
  vtx/.style = {circle, fill=black, inner sep=1.2pt},
  edg/.style = {line width=0.5pt}
}

\newcommand{\rooksCurved}[1][0.9]{%
\begin{tikzpicture}[scale=0.65, every node/.style={font=\small,transform shape}]
  \definecolor{darknavy}{RGB}{0,34,102}
  \tikzset{
    vtx/.style={circle,fill=darknavy,inner sep=0pt,minimum size=4pt},
    edg/.style={line width=0.6pt,darknavy,line cap=round,shorten >=1.3pt,shorten <=1.3pt},
    outline/.style={line width=0.8pt,darknavy,line cap=round,shorten >=1.3pt,shorten <=1.3pt}
  }

  \begin{scope}[shift={(-3.6,0)}]
    \coordinate (u1) at (-2.0,  0.8);
    \coordinate (u2) at (-2.0, -0.8);
    \coordinate (w1) at ( 2.0,  1.2);
    \coordinate (w2) at ( 2.0,  0.0);
    \coordinate (w3) at ( 2.0, -1.2);

    \draw[edg] (u1)--(w1) (u1)--(w2) (u1)--(w3);
    \draw[edg] (u2)--(w1) (u2)--(w2) (u2)--(w3);

    \foreach \P in {u1,u2,w1,w2,w3}{\node[vtx] at (\P) {};}
    \node[left=4pt]  at (u1) {$u_1$};
    \node[left=4pt]  at (u2) {$u_2$};
    \node[right=4pt] at (w1) {$w_1$};
    \node[right=4pt] at (w2) {$w_2$};
    \node[right=4pt] at (w3) {$w_3$};
  \end{scope}

  \begin{scope}[shift={(3.6,0)}]
    \coordinate (r11) at (-1.8,  0.8);
    \coordinate (r12) at ( 0.0,  0.8);
    \coordinate (r13) at ( 1.8,  0.8);
    \coordinate (r21) at (-1.8, -0.8);
    \coordinate (r22) at ( 0.0, -0.8);
    \coordinate (r23) at ( 1.8, -0.8);

    \draw[outline,opacity=0.3] (r11)--(r12)--(r13)--(r23)--(r22)--(r21)--cycle;

    \draw[edg] (r11)--(r12) (r12)--(r13);
    \draw[edg] (r21)--(r22) (r22)--(r23);

    \draw[edg] (r11)--(r21) (r12)--(r22) (r13)--(r23);

    \draw[edg] (r11) .. controls ($ (r12) + (-0.6,-#1) $) and ($ (r12) + (0.6,-#1) $) .. (r13);
    \draw[edg] (r21) .. controls ($ (r22) + (-0.6,#1) $)  and ($ (r22) + (0.6,#1) $)  .. (r23);

    \foreach \P in {r11,r12,r13,r21,r22,r23}{\node[vtx] at (\P) {};}

    \node[above=4pt] at (r11) {$\{u_1,w_1\}$};
    \node[above=4pt] at (r12) {$\{u_1,w_2\}$};
    \node[above=4pt] at (r13) {$\{u_1,w_3\}$};
    \node[below=4pt] at (r21) {$\{u_2,w_1\}$};
    \node[below=4pt] at (r22) {$\{u_2,w_2\}$};
    \node[below=4pt] at (r23) {$\{u_2,w_3\}$};
  \end{scope}
\end{tikzpicture}%
}

\newcommand{\qthreeLine}{%
\begin{tikzpicture}[scale=0.65, every node/.style={font=\small,transform shape}]
  \definecolor{darknavy}{RGB}{0,34,102}
  \tikzset{
    vtx/.style={circle,fill=darknavy,inner sep=0pt,minimum size=4pt},
    edg/.style={line width=0.6pt,darknavy,line cap=round,shorten >=1.3pt,shorten <=1.3pt},
    outline/.style={line width=0.8pt,darknavy,line cap=round,shorten >=1.3pt,shorten <=1.3pt}
  }

  \begin{scope}[shift={(-3.2,0)}]
    \coordinate (F1) at (-1.6,-1.2);
    \coordinate (F2) at ( 0.6,-1.2);
    \coordinate (F3) at ( 0.6, 1.2);
    \coordinate (F4) at (-1.6, 1.2);

    \coordinate (B1) at ( 0.0,-0.2);
    \coordinate (B2) at ( 2.2,-0.2);
    \coordinate (B3) at ( 2.2, 2.2);
    \coordinate (B4) at ( 0.0, 2.2);

    \draw[outline] (F1)--(F2)--(F3)--(F4)--cycle;
    \draw[outline] (B1)--(B2)--(B3)--(B4)--cycle;
    \draw[edg] (F1)--(B1) (F2)--(B2) (F3)--(B3) (F4)--(B4);

    \foreach \P in {F1,F2,F3,F4,B1,B2,B3,B4}{\node[vtx] at (\P) {};}
  \end{scope}

  \begin{scope}[shift={(3.2,0)}]
    \coordinate (mf12) at ({(-1.6+0.6)/2},  -1.2);
    \coordinate (mf23) at ( 0.6,              {(-1.2+1.2)/2});
    \coordinate (mf34) at ({(-1.6+0.6)/2},    1.2);
    \coordinate (mf41) at (-1.6,              {(-1.2+1.2)/2});

    \coordinate (mb12) at ({(0.0+2.2)/2},    -0.2);
    \coordinate (mb23) at ( 2.2,              {( -0.2+2.2)/2});
    \coordinate (mb34) at ({(0.0+2.2)/2},     2.2);
    \coordinate (mb41) at ( 0.0,              {( -0.2+2.2)/2});

    \coordinate (m1) at ({(-1.6+0.0)/2},  {(-1.2-0.2)/2});
    \coordinate (m2) at ({( 0.6+2.2)/2},  {(-1.2-0.2)/2});
    \coordinate (m3) at ({( 0.6+2.2)/2},  {( 1.2+2.2)/2});
    \coordinate (m4) at ({(-1.6+0.0)/2},  {( 1.2+2.2)/2});

    \draw[outline] (mf12)--(m2)--(mb23)--(mb34)--(m4)--(mf34)--(mf41)--(m1)--(mb41)--(mb12)--cycle;

    \draw[edg] (mf12)--(mf41)--(m1)--cycle;
    \draw[edg] (mf12)--(mf23)--(m2)--cycle;
    \draw[edg] (mf23)--(mf34)--(m3)--cycle;
    \draw[edg] (mf34)--(mf41)--(m4)--cycle;

    \draw[edg] (mb12)--(mb41)--(m1)--cycle;
    \draw[edg] (mb12)--(mb23)--(m2)--cycle;
    \draw[edg] (mb23)--(mb34)--(m3)--cycle;
    \draw[edg] (mb34)--(mb41)--(m4)--cycle;

    \draw[edg] (mf12)--(mf23)--(mf34)--(mf41)--cycle;
    \draw[edg] (mb12)--(mb23)--(mb34)--(mb41)--cycle;
    \draw[edg] (mf12)--(m1)--(mb12)--(m2)--cycle;
    \draw[edg] (mf23)--(m2)--(mb23)--(m3)--cycle;
    \draw[edg] (mf34)--(m3)--(mb34)--(m4)--cycle;
    \draw[edg] (mf41)--(m4)--(mb41)--(m1)--cycle;

    \foreach \P in {mf12,mf23,mf34,mf41,mb12,mb23,mb34,mb41,m1,m2,m3,m4}{\node[vtx] at (\P) {};}
  \end{scope}
\end{tikzpicture}%
}

\newcommand{\kfiveLine}{%
\begin{tikzpicture}[scale=0.65, every node/.style={font=\small,transform shape}]
  \definecolor{darknavy}{RGB}{0,34,102}
  \tikzset{
    vtx/.style={circle,fill=darknavy,inner sep=0pt,minimum size=4pt},
    edg/.style={line width=0.6pt,darknavy,line cap=round,shorten >=1.3pt,shorten <=1.3pt},
    outline/.style={line width=0.8pt,darknavy,line cap=round,shorten >=1.3pt,shorten <=1.3pt}
  }

  \begin{scope}[shift={(-3.6,0)}]
    \def\Rk{1.9}
    \foreach \name/\ang in {P1/90, P2/18, P3/-54, P4/-126, P5/162}{
      \coordinate (\name) at ({\Rk*cos(\ang)},{\Rk*sin(\ang)});
    }

    \draw[outline] (P1)--(P2)--(P3)--(P4)--(P5)--cycle;

    \draw[edg] (P1)--(P3) (P1)--(P4) (P1)--(P5);
    \draw[edg] (P2)--(P4) (P2)--(P5);
    \draw[edg] (P3)--(P5);

    \foreach \P in {P1,P2,P3,P4,P5}{ \node[vtx] at (\P) {}; }

    \node at ({(\Rk+0.35)*cos(90)},   {(\Rk+0.35)*sin(90)})   {$1$};
    \node at ({(\Rk+0.35)*cos(18)},   {(\Rk+0.35)*sin(18)})   {$2$};
    \node at ({(\Rk+0.35)*cos(-54)},  {(\Rk+0.35)*sin(-54)})  {$3$};
    \node at ({(\Rk+0.35)*cos(-126)}, {(\Rk+0.35)*sin(-126)}) {$4$};
    \node at ({(\Rk+0.35)*cos(162)},  {(\Rk+0.35)*sin(162)})  {$5$};
  \end{scope}

  \begin{scope}[shift={(3.6,0)}]
    \def\R{2.2}
    \foreach \name/\ang in {
      N12/90, N13/54, N23/18, N24/-18, N34/-54,
      N35/-90, N45/-126, N14/-162, N15/162, N25/126
    }{
      \coordinate (\name) at ({\R*cos(\ang)},{\R*sin(\ang)});
    }

    \draw[outline]
      (N12)--(N13)--(N23)--(N24)--(N34)--(N35)--
      (N45)--(N14)--(N15)--(N25)--cycle;

    \draw[edg] (N12)--(N14) (N12)--(N15) (N12)--(N24) (N12)--(N25);
    \draw[edg] (N13)--(N14) (N13)--(N15) (N13)--(N34) (N13)--(N35);
    \draw[edg] (N23)--(N34) (N23)--(N35) (N23)--(N25);
    \draw[edg] (N24)--(N45) (N24)--(N25);
    \draw[edg] (N34)--(N45) (N34)--(N14);
    \draw[edg] (N35)--(N14) (N35)--(N15);
    \draw[edg] (N45)--(N14) (N45)--(N25);
    \draw[edg] (N14)--(N15);
    \draw[edg] (N15)--(N25);

    \foreach \P in {N12,N13,N23,N24,N34,N35,N45,N14,N15,N25}{
      \node[vtx] at (\P) {};
    }

    \node at ({(\R+0.4)*cos(90)},  {(\R+0.4)*sin(90)})   {$\{1,2\}$};
    \node at ({(\R+0.4)*cos(54)},  {(\R+0.4)*sin(54)})   {$\{1,3\}$};
    \node at ({(\R+0.55)*cos(18)}, {(\R+0.55)*sin(18)})  {$\{2,3\}$};
    \node at ({(\R+0.55)*cos(-18)},{(\R+0.55)*sin(-18)}) {$\{2,4\}$};
    \node at ({(\R+0.4)*cos(-54)}, {(\R+0.4)*sin(-54)})  {$\{3,4\}$};
    \node at ({(\R+0.4)*cos(-90)}, {(\R+0.4)*sin(-90)})  {$\{3,5\}$};
    \node at ({(\R+0.4)*cos(-126)},{(\R+0.4)*sin(-126)}) {$\{4,5\}$};
    \node at ({(\R+0.55)*cos(-162)},{(\R+0.55)*sin(-162)}){$\{1,4\}$};
    \node at ({(\R+0.55)*cos(162)}, {(\R+0.55)*sin(162)}) {$\{1,5\}$};
    \node at ({(\R+0.4)*cos(126)},  {(\R+0.4)*sin(126)})  {$\{2,5\}$};
  \end{scope}
\end{tikzpicture}%
}

\usepackage{float}
\begin{document}

\title{A Graph-Theoretic Approach to Quantum Measurement Incompatibility}
\author{Daniel McNulty}
\affiliation{Dipartimento di Fisica, Università di Bari, 70126 Bari, Italy}
\date{November 20, 2025}

\begin{abstract}
Measurement incompatibility---the impossibility of jointly measuring certain quantum observables---is a fundamental resource for quantum information processing. We develop a graph-theoretic framework for quantifying this resource for large families of binary measurements, including Pauli observables on multi-qubit systems and $k$-body Majorana observables on $n$-mode fermionic systems. To each set of observables we associate an anti-commutativity graph, whose vertices represent observables and whose edges indicate pairs that anti-commute. In this representation, the incompatibility robustness---the minimal amount of classical noise required to render the set jointly measurable---becomes a graph invariant. We derive general bounds on this invariant in terms of the Lovász number, clique number, and fractional chromatic number, and show that the Lovász number yields the correct asymptotic scaling for $k$-body Majorana observables. For line graphs $L(G)$, which naturally arise in the characterisation of exactly solvable spin models, we obtain spectral bounds on the robustness expressed through the energy and skew-energy of the underlying graph $G$. These bounds become tight for highly symmetric graphs, leading to closed formulas for several graph families. Finally, we identify structural conditions under which the robustness is determined by a simple function of the graph’s maximum degree and the number of vertices and edges, and show that such extremal cases occur only when combinatorial structures such as Hadamard, conference or weighing matrices exist.
\end{abstract}

\maketitle
 
 \section{Introduction} 

Many of the remarkable consequences of quantum theory—including uncertainty relations, non-locality, contextuality, and complementarity—originate from the non-commutativity of certain quantum observables \cite{wolf09,xu19,busch04}. A basic implication of non-commutativity is the impossibility of simultaneously measuring certain physical properties of a quantum system, which precludes the existence of a joint probability distribution with marginals matching the statistics of each observable. This lack of joint measurability, known as \emph{measurement incompatibility}, imposes fundamental limits on how information can be transmitted and retrieved \cite{guhne23,mcnulty22}. At the same time, it serves as a quantum resource \cite{heinosaari15}, underpinning applications such as quantum steering \cite{quintino14,uola15} and quantum state discrimination \cite{skrzypczyk19}.

Non-commutativity, however, is not by itself sufficient for incompatibility: certain non-commuting POVMs admit a joint measurement, for instance when they are sufficiently unsharp. A natural way to quantify incompatibility is to mix the measurements with noise and determine the minimal noise level required to render them jointly measurable—the \emph{robustness of incompatibility}. Quantifying measurement incompatibility in this way is useful for a number of purposes. Beyond serving as a resource-theoretic measure of non-classicality \cite{heinosaari15}, robustness measures provide insights into the quantum-to-classical transition in open quantum systems \cite{kiukas22,kiukas23}, and determine bounds on the sample complexity of estimation tasks in quantum computing \cite{mcnulty22,mcnulty24,majsak24}. For highly symmetric measurement sets---such as mutually unbiased bases \cite{durt10,mcnulty24mubs}---robustness measures can sometimes be computed analytically \cite{designolle19b,carmeli19}. In general, however, it requires semidefinite programming, which quickly becomes intractable as the number of measurements or dimension of the system grows \cite{designolle19a,cavalcanti17}.

In this work, we develop a graph--theoretic framework that provides analytic and scalable methods for quantifying measurement incompatibility for collections of observables defined by their (anti-)commutation relations. This complements and extends previous studies on the incompatibility of binary quantum observables \cite{busch86,brougham07,kunjwal14,bluhm18,mcnulty24,majsak24}. Graph--theoretic methods have previously shed light on a variety of fundamental and applied aspects of quantum mechanics, ranging from uncertainty relations \cite{gois23,xu23,moran24} and quantum correlations \cite{cabello10,cabello14,kunjwal19}, to dimension witnessing \cite{ray21} and the study of quantum channels \cite{duan12}. A notable example is the Lovász number, originally introduced as an upper bound on the Shannon capacity of a graph \cite{lovasz79}, which has since found important applications in the study of quantum non--contextual inequalities \cite{cabello10,cabello14} and uncertainty relations \cite{gois23,xu23,moran24}. Similar graph invariants will play a role in this work.

Our framework associates each vertex of a graph with a binary quantum observable and uses the edge set to encode their (anti-)commutation relations. In particular, the \emph{anti-commutativity graph} $G$ of a set of $n$ observables is the graph on $n$ vertices in which two vertices are adjacent if and only if the corresponding operators anti-commute, while non-adjacent vertices correspond to commuting ones. Examples of observables that can be described in this way include tensor products of Pauli operators---ubiquitous throughout quantum computing and quantum information---as well as many-body fermionic Majorana observables. Representing observables by anti-commutativity graphs (also known as frustration graphs) has proved useful in several contexts, including quantum spin systems \cite{chapman20,chapman22,chapman23}, classifying Pauli Lie algebras \cite{aguilar24}, bounding uncertainties and the joint numerical range of Pauli strings \cite{xu23,gois23,moran24}, as well as in the analysis of variational quantum algorithms \cite{kirby19,gokhale19}.

We first show that for any simple connected graph $G$ there exists a set of observables whose anti-commutation relations are described by $G$ (Prop.~\ref{prop:arb_graphs}). Moreover, we prove that the incompatibility robustness depends only on $G$, and is therefore completely determined by the anti-commutation relations rather than by any particular realisation of the observables (Prop.~\ref{prop:eta_invariance}). We then establish an upper bound on the incompatibility robustness in terms of the Lovász number of the anti-commutativity graph (Theorem~\ref{thm:lovasz}). This bound is saturated, for example, by the complete graph, which corresponds to any set of pairwise anti-commuting observables. In some cases, tighter bounds can be obtained from the clique number by considering induced subgraphs (Prop.~\ref{prop:clique}). We also derive a lower bound by constructing joint measurements based on graph colourings that partition the observables into commuting subsets. In particular, by randomising over POVMs that measure each coloured subset, we obtain a lower bound in terms of the fractional chromatic number of the graph (Prop.~\ref{prop:frac}).

A central part of our analysis focuses on the special case where the anti-commutation pattern among observables is represented by a line graph: a graph obtained by replacing each edge of an underlying ``root'' graph with a vertex, and connecting two vertices whenever their corresponding edges in the root graph are incident. Importantly, line graphs can be realised by a suitable family of quadratic Majorana observables (Prop.~\ref{prop:quadratic-line}) This relation was recently used to characterise when many-body spin--$1/2$ systems can be mapped to exactly solvable free-fermion systems \cite{chapman20,chapman22,chapman23}. For these graphs, we derive bounds on the incompatibility robustness in terms of the spectral properties of the root graph, namely its energy and skew--energy (Theorem~\ref{thm:spectra}). These results yield exact formulas for several families, including cycles, Johnson graphs, hypercubes and rook’s graphs (Cor.~\ref{prop:eta-cycle}--\ref{cor:eta-LKmn}), as well as tight asymptotic bounds for paths (Cor.~\ref{prop:eta_paths}). Furthermore, we determine that a line graph saturates a universal upper bound only if it is regular and its root graph admits an orientation whose skew-adjacency matrix is a skew-weighing matrix (Prop.~\ref{prop:weighing}), a well-studied object in combinatorics \cite{colbourn10}. This establishes a connection between joint measurability and classical problems in algebraic combinatorics, such as the existence of Hadamard and skew-conference matrices.

Finally, we turn to a class of observables that arise in many-body fermionic systems and have applications in quantum chemistry and partial tomography \cite{babbush23,clinton24,zhao21,wan23}: degree--$k$ Majorana monomials. The anti-commutativity graphs of these observables are examples of merged Johnson graphs \cite{jones05}. We employ recent tight asymptotic bounds on the Lovász number of these graphs (established by Linz \cite{linz24}) to characterise their measurement incompatibility. In the large-system limit, with fixed $k$, we show that the Lovász number provides a tight asymptotic bound on the incompatibility robustness (Prop.~\ref{prop:inc_rob_merged}).

The paper is organised as follows. Section~\ref{sec:preliminaries} introduces basic definitions and notation. In Section~\ref{sec:arbitrary_graphs} we show that any pattern of anti-commutation relations can be physically realised by binary observables, and that the incompatibility robustness does not depend on the specific choice of observables realising the pattern. General upper and lower bounds on the incompatibility robustness are derived in Section~\ref{sec:lovasz}. Sections~\ref{sec:lines} and \ref{sec:fermionic} consider line graphs and merged Johnson graphs, respectively. We conclude in Section~\ref{sec:summary} with a summary and outlook.

\section{Preliminaries}\label{sec:preliminaries}

\subsection{Quantum measurements}

A quantum measurement on a Hilbert space $\mathcal{H}$ with finite outcome set $\Omega$ is described by a \emph{positive operator-valued measure} (POVM) $\M=\{\M(a)\}_{a\in\Omega}$, i.e., a collection of positive semidefinite operators $\M(a)\in\mathcal{B}(\mathcal{H})$ satisfying
\begin{equation}
\M(a)\ge 0, \qquad \sum_{a\in\Omega}\M(a)=\id,
\end{equation}
where $\mathcal{B}(\mathcal{H})$ denotes the space of bounded linear operators on $\mathcal{H}$. Each operator $\M(a)$ is the POVM effect associated with outcome $a$, 
and the probability of obtaining $a$ when measuring a quantum state $\rho$ is given by the Born rule $\mathbb{P}(a|\rho)=\tr{\rho\,\M(a)}$. 

A \emph{projective measurement} (PVM) is a special instance of a POVM in which all effects are orthogonal projections, i.e., $\M(a)\M(a')=\delta_{aa'}\M(a)$. Such a measurement corresponds to a Hermitian operator $A$ (the \emph{observable}) on $\mathcal H$ with spectral decomposition $A = \sum_{a\in \Omega} a\, \M(a)$, where $\M(a)$ is the eigenprojection associated with eigenvalue $a$. Throughout, we call projective measurements (and their corresponding observables) \emph{sharp}, 
and general non-projective POVMs \emph{unsharp}.

In this work, we focus on families of binary observables
\begin{equation}\label{eq:binary_observables}
\mathcal{A}=\{A_v\,|\,v\in V\},
\end{equation}
on a finite dimensional Hilbert space $\mathcal{H}\cong\mathbb{C}^d$, where $V=[n]:=\{1,2,\ldots,n\}$ labels the observables in the family. Each $A_v$ is Hermitian and we assume $A_v^2=\id$. The latter constraint implies that $A_v$ is unitary with eigenvalues $\pm1$, corresponding to a sharp, dichotomic measurement. The associated two-outcome projective measurement has effects
\begin{equation}\label{eq:sharp}
    \M_v(\pm)=\tfrac{1}{2}(\id\pm A_v).
\end{equation}

\subsection{Anti-commutativity graphs}

The algebraic relations among the observables of $\mathcal{A}$ will be represented graphically in the following way. We associate each vertex of a graph with an observable and use the edge set to encode their (anti-)commutation relations. 

\begin{definition}\label{def:anti_graph}
The \emph{anti-commutativity graph} $G = (V,E)$ of a set of observables $\mathcal{A}$ is defined by
\begin{equation}
\{v,v'\}\in E \iff A_vA_{v'}=-A_{v'}A_v,
\end{equation}
so that adjacent vertices correspond to anti-commuting observables, and non-adjacent vertices correspond to commuting ones.
\end{definition}

Two notable examples of observables that admit anti-commutativity graphs are Pauli strings and many-body fermionic Majorana observables. On $m$ qubits, Pauli strings
\begin{equation}\label{eq:paulis}
P=\bigotimes_{i=1}^m P_i,\qquad P_i\in\{\id,\sigma_x,\sigma_y,\sigma_z\},
\end{equation}
anti-commute if they differ on an odd number of sites by single-qubit Paulis that anti-commute.  
Fermionic Majorana observables are built from Majorana operators $\Gamma_1,\ldots,\Gamma_m$ satisfying
\begin{equation}\label{eq:majoranas}
\Gamma_j^\dagger=\Gamma_j,\qquad \Gamma_j^2=\id,\qquad \{\Gamma_j,\Gamma_k\}=2\delta_{jk}\id,
\end{equation}
where $1\le j,k\le m$. The anti-commutativity graph of these $m$ operators is the complete graph $K_m$, and they generate the complex Clifford algebra $\mathrm{Cl}_m(\mathbb{C})$. They admit an irreducible representation on a Hilbert space $\mathcal{H}\cong\mathbb{C}^{2^{\lfloor m/2\rfloor}}$, which coincides with the fermionic Fock space for $m/2$ modes when $m$ is even. 
In this case, the representation can be realised on $m/2$ qubits via the Jordan--Wigner transformation \cite{bravyi02}.

Products of an even number of distinct Majorana operators, such as the quadratic observables
\begin{equation}\label{eq:quads}
A_{jk}:= i\,\Gamma_j\Gamma_k,\qquad j<k,
\end{equation}
anti-commute if and only if they share an odd number of Majorana operators (see Fig.~\ref{fig:majorana-6}).  Graph representations of Pauli and Majorana observables have been used in various contexts, such as uncertainty relations \cite{guhne23,xu23} and free-fermion spin models \cite{chapman20,chapman22}.

An overview of the graph-theoretic notions used in this work is given in Appendix~\ref{A:prelims}.

\begin{figure}[t]
\centering
\begin{tikzpicture}[scale=0.80] 
  \definecolor{darknavy}{RGB}{0,34,102}
  \tikzset{
    vtx/.style={circle,fill=darknavy,inner sep=0pt,minimum size=4pt}, 
    edg/.style={line width=0.8pt,darknavy,line cap=round},              
    vlab/.style={font=\normalsize}                                    
  }

  \node[vtx] (L)  at (0,0)    {};
  \node[vtx] (R)  at (4,0)    {};
  \node[vtx] (LU) at (-1.6,1.4) {};
  \node[vtx] (LD) at (-1.6,-1.4){};
  \node[vtx] (RU) at (5.6,1.4)  {};
  \node[vtx] (RD) at (5.6,-1.4) {};

  \draw[edg] (L)--(R);
  \draw[edg] (L)--(LU);
  \draw[edg] (L)--(LD);
  \draw[edg] (R)--(RU);
  \draw[edg] (R)--(RD);

  \node[vlab] at ($(L)+(0.8,-0.55)$)  {$-\;\Gamma_{1}\Gamma_{2}\Gamma_{3}\Gamma_{6}$};
  \node[vlab] at ($(R)+(-0.9,0.55)$) {$-\;\Gamma_{1}\Gamma_{4}\Gamma_{5}\Gamma_{7}$};

  \node[vlab] at ($(LU)+(-0.10,0.65)$) {$i\,\Gamma_{2}\Gamma_{8}$};
  \node[vlab] at ($(LD)+(-0.10,-0.70)$) {$i\,\Gamma_{3}\Gamma_{9}$};
  \node[vlab] at ($(RU)+(0.10,0.65)$)  {$i\,\Gamma_{4}\Gamma_{10}$};
  \node[vlab] at ($(RD)+(0.10,-0.70)$) {$i\,\Gamma_{5}\Gamma_{11}$};
\end{tikzpicture}
\caption{Example of a six-vertex anti-commutativity graph realised by quadratic and quartic Majorana monomials. Observables anti-commute if and only if they share one or three common Majorana operators from the set $\{\Gamma_j\,|\,j=1,\ldots,11\}$.}
\label{fig:majorana-6}
\end{figure}
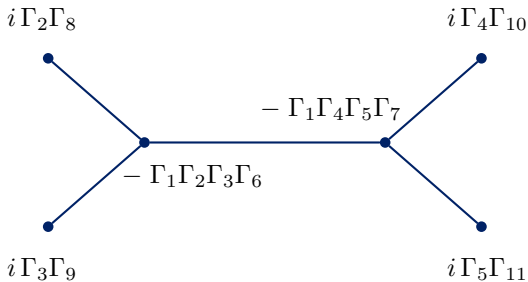

\subsection{Measurement incompatibility}

A collection of measurements are \emph{jointly measurable} if there exists a parent POVM such that each measurement can be recovered as one of its marginals, or equivalently, if their statistics can be obtained from the parent after classical post-processing \cite{ali09,busch16}. For Hermitian observables, this is equivalent to commutativity; hence the family $\mathcal{A}$ in \eqref{eq:binary_observables} cannot be jointly measurable if its anti-commutativity graph contains at least one edge. For general (non-projective) POVMs, commutativity is no longer necessary for joint measurability.

To quantify the degree of incompatibility of $\mathcal{A}$—the extent to which it fails to admit a joint measurement—we define the \emph{incompatibility robustness} \cite{heinosaari15} as
\begin{equation}\label{eq:inc_rob}
\eta(G)\equiv \max\{\eta\in[0,1]\,|\, \mathcal{A}^{\eta}\text{ is jointly measurable}\},
\end{equation}
where
\begin{equation}
\mathcal{A}^\eta = \{A_v^\eta\,|\,v\in V\}, \qquad
A_v^\eta = \eta A_v + \frac{\tr{A_v}}{d}(1-\eta)\id,
\end{equation}
are noisy (unsharp) versions of the original observables with anti-commutativity graph $G$.  
The associated binary POVM elements (cf. \eqref{eq:sharp}) are
\begin{equation}\label{eq:unsharp_povms}
\M_v^\eta(\pm) = \tfrac12(\id \pm \eta A_v), \qquad \eta \in [0,1].
\end{equation}

By definition, when $\eta \leq \eta(G)$ there exists a POVM $\E$ with outcomes $(a_1,\dots,a_n)\in\{\pm1\}^n$ such that each $\M_v^\eta$ is recovered as a marginal:
\begin{equation}\label{eq:postprocessing}
\M_v^\eta(\pm) = \!\!\sum_{\substack{(a_1,\ldots,a_n):\\a_v=\pm 1}} \E(a_1,\ldots,a_n).
\end{equation}
Throughout, we use standard asymptotic notation to describe scaling bounds on the incompatibility robustness $\eta(G)$. For positive functions $f$ and $g$, we write $g=\mathcal{O}(f)$ if there exists $C>0$ such that $g(x)\le C f(x)$ for all sufficiently large $x$; 
$g=\Omega(f)$ if there exists $c>0$ such that $g(x)\ge c f(x)$ for all sufficiently large $x$; and $g=\Theta(f)$ when both relations hold.

\section{Observables from anti-commutativity graphs}\label{sec:arbitrary_graphs}

We now show that \emph{any} simple connected graph $G = (V,E)$ can be realised as the anti-commutativity graph of a set of binary observables $\mathcal{A}$. Disconnected graphs can be treated component-wise, since different components share no incompatibility.

Our construction is based on monomials of the Majorana operators $\{\Gamma_j\}_{j}$ defined in \eqref{eq:majoranas}. For a given graph $G$, we assign a distinct Majorana operator to each edge, and for each vertex $v\in V$ define an observable
\begin{equation}\label{eq:clifford_monomial_main}
A_{\mathscr{I}(v)} := i^{|\mathscr{I}(v)|/2} \prod_{j \in \mathscr{I}(v)} \Gamma_j ,
\end{equation}
where $\mathscr{I}(v)$ is the set of integer labels corresponding to edges incident to $v$.  
If the degree $\text{deg}(v)$ is odd, we include an additional auxiliary operator (labelled uniquely per vertex) so that $|\mathscr{I}(v)|$ is even, ensuring that \eqref{eq:clifford_monomial_main} is Hermitian and squares to~$\id$. See Fig. \ref{fig:majorana-6} for an example construction from a six-vertex anti-commutativity graph.

In this construction, adjacent vertices $v,v'\in V$ share exactly one operator, therefore $A_{\mathscr{I}(v)}$ and $A_{\mathscr{I}(v')}$ anti-commute; non-adjacent vertices share none and therefore commute. These relations follow from the general anti-commutation rule \eqref{eq:car} between arbitrary degree Majorana monomials \cite{bravyi10}. The above argument leads to the following result.

\begin{proposition}\label{prop:arb_graphs}
For every simple connected graph $G = (V,E)$, there exists a set of binary observables whose anti-commutativity graph is $G$.  
The construction~\eqref{eq:clifford_monomial_main} uses at most $|E| + |V|$ Majorana operators, therefore the observables act on a Hilbert space of dimension at most $2^{\lfloor \frac{|V| + |E|}{2} \rfloor}$.
\end{proposition}

The detailed proof is given in Appendix~\ref{A:arbitrary_graphs}. The dimension bound typically exceeds the minimal dimension $2^{r/2}$ required to realise $G$, where $r$ is the rank of its adjacency matrix over $\mathbb{F}_2$ (see Appendix~\ref{A:arbitrary_graphs} for details).

\begin{proposition}\label{prop:eta_invariance}
Let $G=(V,E)$ be a simple graph. The incompatibility robustness $\eta(G)$ depends only on the (anti-)commutation relations encoded by $G$, and is invariant under the particular choice of binary observables $\mathcal{A}=\{A_v\,|\,v\in V\}$ realising them.
\end{proposition}

\begin{proof}
    Associate with each graph $G$ a complex algebra $\mathfrak{C}(G)$ generated by elements $\{x_u : u \in V\}$ which satisfy the relations:
\begin{equation*}
  x_u^2 = 1, \quad x_u=x_u^{*},\quad
  x_u x_w = 
  \begin{cases}
    -x_w x_u & \text{if } \{u, w\} \in E, \\
    \;\;x_w x_u & \text{otherwise}\,,
  \end{cases}
\end{equation*}
where $*$ denotes the adjoint. This algebra is known as a \emph{Clifford graph algebra}, or \emph{quasi-Clifford algebra} \cite{khovanova08,cuypers21,hastings21,xu23}. For example, if $G$ is the complete graph $K_n$, we recover the usual Clifford algebra $\mathfrak{C}(K_n)=\mathrm{Cl}_n(\mathbb{C})$. Suppose that $\{A_u\,|\,u\in V\}$ is the \emph{standard} observable representation of $\mathfrak{C}(G)$, as defined in \cite[Def. 3]{xu23}. The explicit construction of this representation is not essential so we do not reproduce it here. Let $\{A'_u\,|\,u\in V\}$ be a second observable representation. Then, there is a unitary $U$ such that $UA'_uU^\dagger=A_u\otimes D_u$, where $\{D_u\,|\,u\in V\}$ is a set of commuting self-adjoint unitaries \cite{samoilenko12} (see also \cite[Thm. 2]{xu23}). Since $\{D_u\,|\,u\in V\}$ are commuting, $\{A_u\,|\,u\in V\}$ is jointly measurable if and only if $\{A_u\otimes D_u\,|\,u\in V\}$ is jointly measurable if and only if $\{A'_u\,|\,u\in V\}$ is jointly measurable. The latter equivalence holds since a unitary transformation on a set of observables does not affect joint measurability. It follows that the incompatibility robustness is invariant under the choice of observables realising the graph.
\end{proof}

A simple consequence of Prop. \ref{prop:eta_invariance} is that we can delete any vertex that contains a twin without affecting the incompatibility robustness. Vertices $v,v'\in V$ of a graph $G$ are called \emph{twins} if they are non-adjacent and share the same neighbours, i.e., $v\nsim v'$ and $\{u,v\}\in E$ if and only if $\{u,v'\}\in E$ for all $u\in V$.

\begin{corollary}\label{lem:twin}
Let $G$ be the anti-commutativity graph of a set of binary observables $\mathcal{A}$. If $v,v'$ are twins in $G$, then
\begin{equation}
\eta(G\setminus\{v\}) \;=\; \eta(G).
\end{equation}
\end{corollary}

\begin{proof}
By Prop.~\ref{prop:eta_invariance}, $\eta(G)$ depends only the anti-commutation
relations and not on a particular realisation of the observables. Since $u,v$ are twins, they anti-commute with exactly the same observables and commute with each other. We may therefore choose a realisation with
$A_v = A_{v'}$, keeping all other observables unchanged. This preserves all
(anti-)commutation relations. Removing the duplicated observable does not
affect joint measurability, hence $\eta(G\setminus\{v\}) = \eta(G)$.
\end{proof}

For example, the graph $G$ in Fig.~\ref{fig:majorana-6} contains two pairs of twin vertices, such as the pair corresponding to $\{i\Gamma_2\Gamma_8,\,i\Gamma_3\Gamma_9\}$. 
Removing one vertex from each twin pair yields the path graph $P_4$, hence $\eta(G)=\eta(P_4)$.

\section{Bounds on the incompatibility robustness}\label{sec:lovasz}

We now bound the incompatibility robustness \eqref{eq:inc_rob} by graph parameters of the anti-commutativtiy graph. We begin with an upper bound in terms of the Lovász number $\vartheta(G)$. This graph invariant was originally introduced as an efficiently computable bound on the Shannon capacity of a graph \cite{lovasz79} and is a semidefinite programming relaxation of the independence number (see Appendix \ref{A:prelims}).

\begin{theorem}\label{thm:lovasz}
Let $G=(V,E)$ be the anti-commutativity graph of $\mathcal{A}=\{A_v : v\in V\}$, and let $\vartheta(G)$ be its Lovász number. Then
\begin{equation}\label{eq:lovasz_bound}
\eta(G) \;\le\; \sqrt{\vartheta(G)/|V|}\,.
\end{equation}
\end{theorem}

\begin{proof}
For a graph $G$, let $a:V\rightarrow \{\pm 1\}$ be a function assigning a sign $a_v\in\{\pm 1\}$ to each vertex. One can show (see Prop.~\ref{Aprop:inc_rob_bound} in Appendix~\ref{A:theorem_proof}) that for any set of $|V|$ binary observables with anti-commutativity graph $G$,
\begin{equation}\label{eq:eta_norm_bound}
\eta(G) \;\le\; |V|^{-1} \max_{a:V\rightarrow \{\pm 1\}}\big\| \sum_{v\in V} a_v A_v \big\|_{\infty}\,,
\end{equation}
where $\no{\cdot}_{\infty}$ denotes the operator norm. By relaxing the constraint on the coefficients, we now define $b:V\to\mathbb{R}$ and introduce the quantity defined in \cite{hastings21}:
\begin{align}\label{eq:odonnell}
\Psi(G)\;:&=\;\max_{\;b:V\to\mathbb{R},\;\|b\|_2=1}
\bigl\|\sum_{v\in V} b_v A_v\bigr\|_{\infty}^2\\
&\;=\;\max_{\;a:V\to\mathbb{R},\;\|a\|_2=\sqrt{|V|}}
\bigl\|\sum_{v\in V} \frac{a_v}{\sqrt{|V|}} A_v\bigr\|_{\infty}^2\,,\label{eq:odonnell2}
\end{align}
where, $\no{\cdot}_2$ denotes the $\ell^2$--norm on $\mathbb{R}^{|V|}$. Eq. \eqref{eq:odonnell2} follows from the rescaling $b_v=a_v/\sqrt{|V|}$, which enforces $\no{a}_2 = \sqrt{|V|}$. Since the maximisation in \eqref{eq:odonnell2} extends over a continuous set of real coefficients, whereas \eqref{eq:eta_norm_bound} is restricted to signings, we obtain the inequality $\eta(G)|V|\leq\sqrt{\Psi(G)|V|}$, or equivalently $\eta(G)\leq \sqrt{\Psi(G)/|V|}$. Applying the bound $\Psi(G) \le \vartheta(G)$ established by Hastings and O’Donnell~\cite{hastings21} (see Prop.~\ref{lem:hastings} in Appendix~\ref{A:theorem_proof}) yields the claimed result~\eqref{eq:lovasz_bound}.
\end{proof}

The quantity $\Psi(G)$ in \eqref{eq:odonnell} also appears in recent graph-theoretic bounds on the joint numerical range $J(G)=\{(\langle P_1\rangle_\rho,\dots,\langle P_{n}\rangle_\rho):\rho\in\mathcal S(\mathcal H)\}$ of collections of Pauli strings $\{P_j\,|\,j=1,\ldots,n\}$, where $G$ denotes their anti-commutativity graph \cite{xu23}. Defining $r^2:=\sup_{\rho}\sum_{j}\langle P_i\rangle_\rho^2$ as the squared radius of the smallest Euclidean ball containing the joint numerical range, one finds that $\Psi(G)=r^2$.

The following result shows that Theorem~\ref{thm:lovasz} is tight for complete graphs, i.e., when all observables in $\mathcal{A}$ pairwise anti-commute.
\begin{corollary}\label{cor:complete}
Let $G=K_n$ be the complete graph. Then, $\eta(K_n)=n^{-1/2}$.
\end{corollary}
\begin{proof}
Since $\vartheta(K_n)=1$, Thm. \ref{thm:lovasz} gives $\eta(K_n)\leq 1/\sqrt{n}$. A parent POVM achieving this bound was constructed in \cite{mcnulty24}, with effects:
\begin{equation}
\E(a_1\ldots,a_n)=2^{-n}\Big(\id+n^{-\frac{1}{2}}\sum_{v=1}^na_vA_v\Big)\,,
\end{equation}
which are positive semidefinite for all $a_v\in\{\pm 1\}$. Summing over all outcomes except $a_u$ gives $\sum_{v\neq u}\sum_{a_v}\E(a_1,\ldots,a_n)=\M^{\eta}_{u}(a_{u})$, where $\M^\eta_u$ is defined in Eq.~\eqref{eq:unsharp_povms} with $\eta = 1/\sqrt{n}$. Thus the bound in Thm.~\ref{thm:lovasz} is saturated.
\end{proof}

We will see in Sec. \ref{sec:fermionic} that Theorem \ref{thm:lovasz} gives a tight asymptotic bound for merged Johnson graphs.

The bound in Theorem~\ref{thm:lovasz} can sometimes be improved by considering an induced subgraph.  
Given $G = (V,E)$ and $S \subset V$, the \emph{induced subgraph} $G[S]$ has vertex set $S$ and edge set consisting of all edges of $G$ between vertices in $S$.

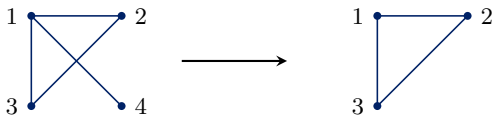
\begin{figure}[!t]
\centering
\begin{tikzpicture}[scale=1, every node/.style={font=\small}]
  \definecolor{darknavy}{RGB}{0,34,102}
  \tikzset{
    vtx/.style={circle,fill=darknavy,inner sep=0pt,minimum size=3pt},
    edg/.style={line width=0.6pt, darknavy}
  }

  \begin{scope}[shift={(0,0)}]
    \coordinate (tl) at (0,1.2);
    \coordinate (tr) at (1.2,1.2);
    \coordinate (br) at (1.2,0);
    \coordinate (bl) at (0,0);
    \draw[edg] (tl)--(tr)--(bl)--cycle;
    \draw[edg] (tl)--(br);
    \node[vtx,label=left:{$1$}]  at (tl) {};
    \node[vtx,label=right:{$2$}] at (tr) {};
    \node[vtx,label=left:{$3$}]  at (bl) {};
    \node[vtx,label=right:{$4$}] at (br) {};
  \end{scope}

  \draw[->,>=stealth,thick] (2.0,0.6) -- (3.4,0.6);

  \begin{scope}[shift={(4.6,0)}]
    \coordinate (tl) at (0,1.2);
    \coordinate (tr) at (1.2,1.2);
    \coordinate (bl) at (0,0);
    \draw[edg] (tl)--(tr)--(bl)--cycle;
    \node[vtx,label=left:{$1$}]  at (tl) {};
    \node[vtx,label=right:{$2$}] at (tr) {};
    \node[vtx,label=left:{$3$}]  at (bl) {};
  \end{scope}
\end{tikzpicture}
\caption{\label{fig:subgraph}
Example illustrating that Prop.~\ref{prop:subgraph} can yield a tighter bound than Thm.~\ref{thm:lovasz} via an induced subgraph. Left: the paw graph $G$ on $V=\{1,2,3,4\}$ has $\vartheta(G)=2$, therefore Thm.~\ref{thm:lovasz} gives $\eta(G)\le 1/\sqrt{2}$.  
Right: the induced subgraph $G[S]$ for $S=\{1,2,3\}$ is $K_3$ with $\vartheta(G[S])=1$. Therefore, Prop.~\ref{prop:subgraph} yields $\eta(G)\le 1/\sqrt{3}$, which is tight.}
\end{figure}

\begin{proposition}\label{prop:subgraph}
Let $G[S]$ be an induced subgraph of $G$ with vertex set $S \subset V$. Then
\begin{equation}\label{eq:subgraph}
\eta(G) \;\le\; \eta(G[S]) \;\le\; \sqrt{\vartheta(G[S])/|S|}.
\end{equation}
\end{proposition}

\begin{proof}
Removing observables for a set cannot increase incompatibility. Thus, if $\mathcal{A}^\eta$ admits a parent POVM, then any subset $\widetilde{\mathcal{A}}^\eta \subset \mathcal{A}^\eta$ also does.
Hence $\eta(G) \le \eta(G[S])$.  
Applying Theorem~\ref{thm:lovasz} to $G[S]$ yields~\eqref{eq:subgraph}.
\end{proof}

Fig.~\ref{fig:subgraph} illustrates an example where Prop.~\ref{prop:subgraph} gives a strictly stronger bound than Theorem~\ref{thm:lovasz}.

In view of Prop.~\ref{prop:subgraph}, an effective bound can be obtained by choosing an induced subgraph with a large vertex set and small Lovász number.  
A natural choice is a \emph{clique}—a subset of vertices in which every pair is adjacent—since cliques have Lovász number $\vartheta = 1$.  
Let $\omega(G)$ denote the \emph{clique number} of $G$, i.e., the size of its largest clique.  
This leads to the following bound.

\begin{proposition}\label{prop:clique}
Let $\omega(G)$ be the clique number of $G$. Then, 
\begin{equation}\label{eq:clique}
\eta(G)\leq \omega(G)^{-\frac{1}{2}}.
\end{equation}
\end{proposition}

\begin{proof}
If $G$ contains a clique of size $\omega(G)$, then there exists an induced subgraph $G[S]$ with $|S| = \omega(G)$ that is complete, i.e., $G[S] = K_{\omega(G)}$.  
By Cor.~\ref{cor:complete}, $\eta(G[S]) = \omega(G)^{-1/2}$.  
Applying Prop.~\ref{prop:subgraph} gives the claimed bound.
\end{proof}

Lower bounds on the incompatibility robustness follow from explicit joint measurements. Suppose the vertices can be partitioned into $a$ commuting subsets. Measuring one subset chosen uniformly at random and assigning uniformly random outcomes to all others yields a parent POVM of $\mathcal{A}^\eta$ with $\eta = 1/a$~\cite{oszmaniec17}. Since the minimum number of commuting subsets that partition the vertices equals the \emph{chromatic number} $\chi(G)$, then $\eta(G) \ge 1/\chi(G)$.

\begin{figure}[!t]\label{cycle_5}
\centering
\begin{tikzpicture}[scale=1, every node/.style={font=\small}]
  \definecolor{myred}{RGB}{204,0,0}
  \definecolor{myblue}{RGB}{0,90,200}
  \definecolor{mygreen}{RGB}{0,140,70}
  \definecolor{myyellow}{RGB}{220,170,0}
  \definecolor{mypink}{RGB}{225,0,130}

  \tikzset{
    v/.style   ={circle,draw=black,fill=black,inner sep=0pt,minimum size=3pt},
    e/.style   ={line width=0.6pt,draw=black},
    dot/.style ={circle,inner sep=0pt,minimum size=4pt,draw=none},
  }

  \def\r{1.25}
  \foreach \i/\ang in {1/90, 2/18, 3/-54, 4/-126, 5/162}{
    \coordinate (L\i) at ({\r*cos(\ang)},{\r*sin(\ang)});
    \coordinate (R\i) at ({\r*cos(\ang)+4.2},{\r*sin(\ang)}); 
  }

  \foreach \i/\j in {1/2,2/3,3/4,4/5,5/1}{
    \draw[e] (L\i)--(L\j);
  }
  \node[dot,fill=myred]   at (L1) {};
  \node[dot,fill=myblue]  at (L2) {};
  \node[dot,fill=mygreen] at (L3) {};
  \node[dot,fill=myred]   at (L4) {};
  \node[dot,fill=myblue]  at (L5) {};
  \node[below=8pt] at ($(L3)!0.5!(L4)$) {$\mathbf{3\!:\!1}$ colouring ($\eta=1/3$)};

  \foreach \i/\j in {1/2,2/3,3/4,4/5,5/1}{
    \draw[e] (R\i)--(R\j);
  }
  \def\off{0.14}
  \foreach \i/\ang/\cA/\cB in {
    1/90/myred/mygreen,
    2/18/myblue/myyellow,
    3/-54/mygreen/mypink,
    4/-126/myyellow/myred,
    5/162/mypink/myblue}{
    \path (R\i) coordinate (C);
    \coordinate (tA) at ({\off*cos(\ang+90)},{\off*sin(\ang+90)});
    \coordinate (tB) at ({\off*cos(\ang-90)},{\off*sin(\ang-90)});
    \node[dot,fill=\cA] at ($(C)+(tA)$) {};
    \node[dot,fill=\cB] at ($(C)+(tB)$) {};
  }
  \node[below=8pt] at ($(R3)!0.5!(R4)$) {$\mathbf{5\!:\!2}$ colouring ($\eta=2/5$)};


\end{tikzpicture}
\caption{\label{fig:cycle_5}
Examples of $a\!:\!b$--colourings for the cycle graph $C_5$.  
Left: a $3\!:\!1$ colouring partitions the vertices into three commuting sets (red, blue, green).  
Right: a fractional $5\!:\!2$ colouring assigns two of five colours (red, blue, green, yellow, pink) to each vertex, with each colour defining a commuting set.
A parent POVM implemented by measuring one of the $a$ commuting classes at random jointly measures $\{\M_v^\eta\,|v\in V\}$ with $\eta\in\{1/3,2/5\}$.}
\end{figure}
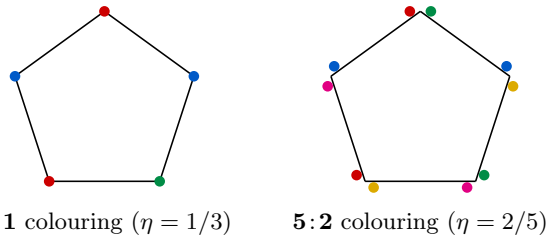

This bound can be improved using \emph{fractional colourings}, in which each vertex is assigned $b$ colours from a palette of $a$, such that the colours assigned to adjacent vertices are distinct. Such an $a\!:\!b$--colouring allows each observable to appear in $b$ of the $a$ commuting sets. By constructing a joint measurement that randomly selects one of the $a$ commuting subsets and measures it, we obtain a parent POVM for $\mathcal{A}^\eta$ with $\eta = b/a$. The optimal ratio is the \emph{fractional chromatic number} $\chi_f(G)$, defined as $\chi_f(G) = \inf_{a,b} a/b$, which gives the following bound.

\begin{proposition}\label{prop:frac}
Let $\chi_f(G)$ be the fractional chromatic number of $G$. Then,
\begin{equation}\label{eq:frac}
\eta(G)\geq 1/\chi_f(G)\,.
\end{equation}
\end{proposition}

Fig.~\ref{fig:cycle_5} illustrates that a fractional colouring can yield a stronger lower bound on~$\eta(G)$ than an ordinary colouring. Note that for vertex-transitive graphs, the fractional chromatic number is $\chi_f(G) = |V| / \alpha(G)$, where $\alpha(G)$ is the independence number.

The chromatic number has been used previously to reduce the measurement cost of estimation schemes in quantum computing algorithms \cite{jena19,yen20,verteletskyi20,gokhale19,bonet20,izmaylov19,clinton24}. A measurement strategy based on fractional colourings was recently applied to classical shadow estimation schemes \cite{king25}. From the perspective of joint measurability, however, such strategies are typically suboptimal and do not yield tight lower bounds. In general, a more refined parent is required to jointly measure a set of observables in a way that approaches the sharpness~$\eta(G)$ (cf.~Sec.~\ref{sec:fermionic}).

\section{Line graphs}\label{sec:lines}

The \emph{line graph} of a (root) graph $G$ is the graph $L(G)$ whose vertex set is the edge set of $G$, and where two vertices of $L(G)$ are adjacent if and only if the corresponding edges of $G$ are incident. For example, the line graph of a complete graph $K_n$ is the Johnson graph $J(n,2)$, as illustrated in Fig. \ref{fig:line}. Line graphs of cycles and paths are cycles and paths, respectively, while the line graphs of a hypercube $Q_d$ and complete bipartite graph $K_{m,n}$ are illustrated in Figs. \ref{fig:LQ3} and \ref{fig:rook}.

\begin{figure}[t!]
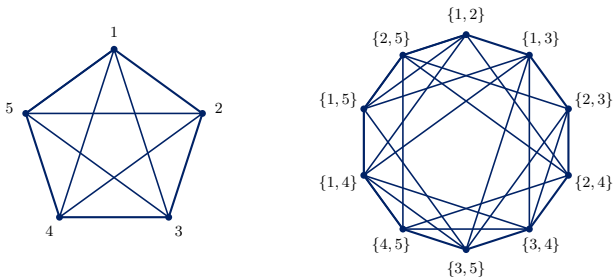

  \kfiveLine
  \caption{The complete graph $K_5$ (left) and its line graph $L(K_5)$ (right) which is the Johnson graph $J(5,2)$. Each edge $\{u,v\}$ of $K_5$ corresponds to a vertex $\{u,v\}$ of $L(K_5)$; two such vertices are adjacent in $L(K_5)$ if and only if the corresponding edges in $K_5$ are incident. }
  \label{fig:line}
\end{figure}

Line graphs $L(G)$ naturally express the anti-commutation relations between quadratic Majorana observables. 
Here, the root graph $G$ serves as the \emph{coupling graph}, whose edges represent pairwise couplings between Majorana operators defined in~\eqref{eq:majoranas}.

\begin{definition}
A (Majorana) coupling graph $G=(V,E)$ consists of a vertex set $V$ labelling Majorana operators $\{\Gamma_v\,|\,v\in V\}$, and an edge set $E$ labelling quadratic observables $\{i\Gamma_u\Gamma_v\,|\,\{u,v\}\in E\}$.
\end{definition}

The line graph $L( G)$ of the coupling graph $G$ is then the anti-commutativity graph of the quadratic observables indexed by $E$: two observables $i\Gamma_u\Gamma_v$ and $i\Gamma_{u'}\Gamma_{v'}$ anti-commute if and only if $\{u,v\}$ and $\{u',v'\}$ share a common vertex in $G$. 

For example, the complete coupling graph $G=K_n$ represents the set of quadratic observables $\{\,i\Gamma_u\Gamma_v : 1\leq u<v\leq n\,\}$ indexed by all $\binom{n}{2}$ edges of $K_n$. The anti-commutativity graph is then the line graph $L(K_n)$, which is the Johnson graph $J(n,2)$ (see Fig.~\ref{fig:line}). Similarly, if the coupling graph is a path, $G=P_n$, the corresponding $n-1$ quadratic observables are $\{i\Gamma_j\Gamma_{j+1} : 1\leq j\leq n-1\}$. Since each observable anti-commutes only with its immediate neighbours, the anti-commutativity graph is again a path, namely $L(P_n)=P_{n-1}$.

\begin{proposition}\label{prop:quadratic-line}
Every anti-commutativity graph that is a line graph can be realised by a family of quadratic Majorana observables. Conversely, every family of quadratic Majorana observables has a line graph as its anti-commutativity graph, provided no repeated observables appear.
\end{proposition}

\begin{proof}
    Let $L(G)$ be a line graph. Its vertices correspond to edges of $G$, hence to quadratic observables $i\Gamma_u\Gamma_v$. Adjacency in $L(G)$ encodes whether two edges share a vertex in $G$, which is equivalent to whether the corresponding quadratic observables anti-commute. Conversely, any set of quadratic Majorana observables (containing no repetitions) defines a coupling graph $G$ whose edges label them, and the anti-commutativity graph is then $L(G)$.
\end{proof}

\subsection{Bounds from the skew-energy}

For a line graph $L(G)$, incompatibility robustness admits bounds in terms of the energy and skew--energy of the graph $G=(V,E)$ and its orientations. 
An \emph{oriented graph} $G^\sigma$ is obtained from $G$ by an orientation $\sigma$ that assigns to each edge $\{u,v\}\in E$ a direction $u\to v$ or $v\to u$. 

The \emph{adjacency matrix} $A(G)$ of $G$ is a $|V|\times|V|$ matrix with entries $A_{uv}=1$ if $\{u,v\}\in E$, and $A_{uv}=0$ otherwise. The \emph{energy} of $G$ is the sum of the absolute values of the eigenvalues of $A(G)$~\cite{gutman01}. The \emph{skew--adjacency matrix} $S(G^\sigma)$ of an orientation $G^\sigma$ is the $|V|\times|V|$ skew--symmetric matrix defined by
\begin{equation}\label{eq:skew-adjacency}
S_{uv} =
\begin{cases}
1 & \text{if } u\to v,\\
-1 & \text{if } v\to u,\\
0 & \text{otherwise}.
\end{cases}
\end{equation}
The \emph{skew--energy} of $G^\sigma$ is the sum of the absolute values of the eigenvalues of $S(G^\sigma)$~\cite{adiga10}.

\begin{definition}\label{def:energies}
Let $G=(V,E)$ be a graph and let $G^\sigma$ be an orientation.
We write $\lambda_j(A)$ and $\lambda_j(S)$ as the eigenvalues of the adjacency matrix $A(G)$ and the skew--adjacency matrix $S(G^\sigma$), respectively. 
The \emph{energy} of $G$ and the \emph{skew--energy} of $G^\sigma$ are defined as
\begin{equation*}
  \mathcal{E}(G) \;=\; \sum_{j}\bigl|\lambda_j(A)\bigr|,
  \qquad
  \mathcal{E}_s(G^\sigma) \;=\; \sum_{j} \bigl|\lambda_j(S)\bigr|\,.
\end{equation*}
The \emph{maximum skew--energy} of $G$ is
\begin{equation}\label{eq:max_skew}
    \mathcal{E}_s^{\max}(G) \;=\; \max_{\sigma} \mathcal{E}_s(G^\sigma)\,,
\end{equation}
where the maximum is taken over all orientations.
\end{definition}

These quantities have been studied extensively in the literature (see, e.g., the surveys \cite{gutman01,li13}). We now bound the incompatibility robustness of a line graph $L(G)$ using the maximal skew-energy of the graph $G$. Before stating the result, recall that a graph $G$ is \emph{edge-transitive} if all edges are equivalent under its automorphism group, i.e., for any two edges $e,e'\in E(G)$ there exists an automorphism mapping $e$ to $e'$ (see Appendix \ref{A:prelims} for further details). 

\begin{theorem}\label{thm:spectra}
Let $L(G)$ be the line graph of a graph $G=(V,E)$ with maximum skew-energy $\mathcal{E}_s^{\max}(G)$. Then
\begin{equation}\label{eq:skew_bound_eigs}
   \eta(L(G))\;\le\; \frac{1}{2|E|}\mathcal{E}_s^{\max}(G)\,.
\end{equation}
If $G$ is edge-transitive then equality holds.
\end{theorem}

The proof, given in Appendix~\ref{sec:A:theorem_proof2}, is based on the bound~\eqref{eq:eta_norm_bound} together with a correspondence between the spectral properties of $\sum_j a_j A_j$ (where $A_j$ are the binary observables) and the skew--adjacency matrix of an oriented graph $G^\sigma$. To prove that equality holds when $G$ is edge-transitive, we show that the associated measurements are \emph{uniformly and rigidly symmetric}
(see Lem.~\ref{lem:uniform_rigid} in Appendix~\ref{A:sym}). These symmetries reduce the dual semidefinite program to a two-parameter ansatz, under which
\eqref{eq:eta_norm_bound} becomes tight (cf.~\cite{nguyen20}).

The maximisation over all orientations $\sigma$ in~\eqref{eq:max_skew} can be simplified by the method of \emph{switching} \cite{adiga10,yaoping11}, which is the process of reversing all edges incident to a chosen vertex. Two orientations $\sigma$ and $\sigma'$ are \emph{switching equivalent} if $G^{\sigma'}$ can be obtained from $G^{\sigma}$ by a sequence of switchings. The equivalence classes under this relation are called \emph{switching classes}.

\begin{proposition}
The maximisation over all orientations in \eqref{eq:skew_bound_eigs} reduces to a maximisation over switching classes. For a connected graph with $n$ vertices and $m$ edges, there are $2^{m-n+1}$ switching classes.
\end{proposition}

\begin{proof}
 Two oriented graphs $G^\sigma$ and $G^{\sigma'}$ of $G$ that are switching equivalent have the same skew-spectrum \cite[Lemma 2.8]{yaoping11}, therefore the maximisation reduces to switching classes. The total number of switching classes is given in \cite[Prop. 3.1]{naserasr15}.
\end{proof}

If the maximum skew-energy is not easily calculated, a lower bound on $\eta(L(G))$ may be obtained from the energy of $G$. For a bipartite graph $G$---whose vertex set can be partitioned into two disjoint subsets with edges only between them---it is known that there exists an orientation $\sigma$ such that the skew-energy $\mathcal{E}_s(G^{\sigma})$ equals the energy $\mathcal{E}(G)$ of the underlying graph $G$ \cite{shader09}. This observation is used to prove the following result.

\begin{proposition}
    Let $L(G)$ be the line graph of a bipartite and edge-transitive graph $G=(V,E)$ with energy $\mathcal{E}(G)$. Then,
    \begin{equation}
        \eta(L(G))\geq \frac{1}{2|E|}\mathcal{E}(G)\,.
    \end{equation}
\end{proposition}

 \begin{proof}
A graph $G$ is bipartite if and only if there is an orientation $\sigma$ such that $\mathcal{E}_s(G^{\sigma})=\mathcal{E}(G)$ \cite{shader09}. If $G$ is edge-transitive, Thm. \ref{thm:spectra} gives the equality $\eta(L(G))=\mathcal{E}_s^{\max}(G)/2|E|$. Since $\mathcal{E}_s^{\max}(G)\geq \mathcal{E}_s(G^{\sigma})=\mathcal{E}(G)$, we get the claimed bound.
 \end{proof}

\subsection{Examples}

We now apply known results on the energy and skew-energy of graphs to obtain simple formulas for the incompatibility robustness of observables whose anti-commutativity graphs are line graphs. Our results cover cycles, Johnson graphs, hypercube line graphs, rook's graphs and paths.

\subsubsection{Cycles}\label{sec:cycle}

A cycle $C_n$---which is a line graph of itself---has $n$ vertices arranged in a single closed loop. Cycles are edge-transitive, vertex-transitive, and 2-regular; moreover $C_n$ is bipartite when $n$ is even. A family of observables whose anti-commutativity graph is $C_n$ is $\{ i\,\Gamma_j\Gamma_{j+1} \;|\; j=1,\ldots,n \}$, where indices are taken modulo $n$. A simple formula for their incompatibility is given by the following result, and is graphically illustrated in Fig. \ref{fig:eta-side-by-side}.

\begin{corollary}\label{prop:eta-cycle}
For cycle graphs $C_n$, incompatibility robustness is given by
 \begin{equation}\label{eq:eta-cycle}
        \eta(C_n)= \begin{cases}
 \frac{1}{n}\,\cot(\pi/2n), & n \text{ odd},\\
\frac{2}{n}\,\csc(\pi/n), & n \text{ even.}
\end{cases}
    \end{equation}
As $n\rightarrow\infty$, $\eta(C_n)\rightarrow\frac{2}{\pi}$.
\end{corollary}

The proof in Appendix \ref{A:line_examples} relies on Theorem \ref{thm:spectra} and calculations of the skew-energies for the two distinct switching classes of oriented cycles, given in \cite{adiga10}. 

When comparing \eqref{eq:eta-cycle} with the bound in Theorem \ref{thm:lovasz}, we see that the Lovász number does not give a tight bound for cycle graphs.

\subsubsection{Johnson graphs}\label{sec:johnson}

The Johnson graph $J(n,2)$ is the line graph of the complete graph $K_n$ (as illustrated in Fig. \ref{fig:line}) which is both vertex- and edge-transitive. The vertex set of $J(n,2)$ is the collection of all $2$-element subsets of $[n]$, denoted by $V=\binom{[n]}{2}$, with each element $e=\{u,v\}\in V$ corresponding to an edge of $K_n$. Two vertices $e,e'\in V$ are adjacent in $J(n,2)$ whenever the corresponding edges of $K_n$ share a common endpoint, i.e., $|e\cap e'|=1$. The Johnson graph, which is edge-transitive, vertex-transitive, and $2(n-2)$-regular, therefore represents the set of quadratic observables $\{\,i\Gamma_u\Gamma_v : 1\leq u<v\leq n\,\}$.

\begin{corollary}\label{cor:johnson}
For Johnson graphs $J(n,2)$, the incompatibility robustness of the corresponding set of $\binom{n}{2}$ observables satisfies
\begin{equation}\label{eq:eta-triangular}
\eta\big(J(n,2)\big) \;\le\; \frac{1}{\sqrt{\,n-1\,}}\,,
\end{equation}
with equality if and only if a skew--conference matrix of order $n$ exists.
\end{corollary}
A \emph{conference matrix} of order $n$ is an $n\times n$ matrix $C$ with zero diagonal and $\pm 1$ off--diagonal entries such that $CC^\top = (n-1)\id$ \cite{colbourn10}. Equivalently, all eigenvalues of $C$ have modulus $\sqrt{n-1}$. A \emph{skew--conference matrix} is a conference matrix that is skew--symmetric, hence it is also a tournament matrix. The set of all $n\times n$ tournament matrices is equivalent to the set of all skew-adjacency matrices of orientated complete graphs $K_n^\sigma$.

The existence of skew-conference matrices is equivalent to the existence of skew-Hadamard matrices. Hence, \eqref{eq:eta-triangular} is saturated if and only if an $n\times n$ skew-Hadamard matrix exists, which is conjectured to be the case whenever $n$ is divisible by four \cite{koukouvinos08}.

The proof of Cor. \ref{cor:johnson}, as shown in Appendix \ref{A:line_examples}, follows from a known bound on the maximum skew-energy of tournament matrices \cite{ito17}. While Cor. \ref{cor:johnson} was first proved in \cite{mcnulty24}, the present framework places it in a more general graph-theoretic setting that connects measurement incompatibility to spectral properties of the underlying graph.

\subsubsection{Hypercube line graphs}\label{sec:hypercube}

\begin{figure}[t!]
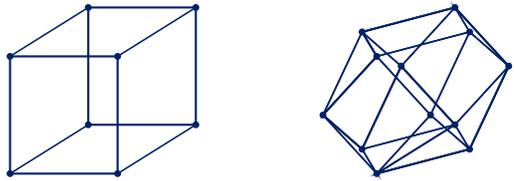

  \centering
  \qthreeLine
  \caption{The 3-cube $Q_3$ (left) and its line graph $L(Q_3)$ (right).}
  \label{fig:LQ3}
\end{figure}

The $d$-dimensional hypercube $Q_d$ is the graph whose vertex set is
$V(Q_d)=\{0,1\}^d$, with two vertices adjacent if and only if they differ in exactly one coordinate. Equivalently, $Q_d$ is the Cartesian product of $d$ copies of $K_2$. It has $|V(Q_d)|=2^d$, $|E(Q_d)|=d2^{d-1}$,
and is $d$-regular, bipartite and edge-transitive. Its line graph, which has the product decomposition $L(Q_d)=\ K_d \,\square\, Q_{d-1}$, has $|V(L(Q_d))|=d\,2^{d-1}$. Since each edge of $Q_d$ is incident to two vertices of degree $d$, each vertex of $L(Q_d)$ has degree $2(d-1)$ and is therefore regular. 

\begin{corollary}\label{cor:eta-hypercube-exact}
Let $d\geq 3$. For the line graph $L(Q_d)$ of the $d$-dimensional hypercube $Q_d$, the incompatibility robustness of the corresponding set of $d2^{d-1}$ observables satisfies
\begin{equation}
\eta\big(L(Q_d)\big)\;=\;\frac{1}{\sqrt{d}}.
\end{equation}
\end{corollary}
The proof follows from a construction in \cite{tian11} of an orientation of $Q_d$ that achieves the optimal skew-energy of any graph (cf. Appendix \ref{A:line_examples}).

\subsubsection{Rook's graphs}\label{sec:rooks}

The line graph of the complete bipartite graph $K_{r,s}$ is the \emph{rook's graph}, whose vertices correspond to the squares of an $r\times s$ chessboard, with edges connecting squares lying in the same row or column (see Fig.~\ref{fig:rook}). Rook's graphs with $rs$ vertices are the Cartesian product of two complete graphs $K_r\square K_s = L(K_{r,s})$. The graph $K_{r,s}$ is regular only if $r=s$, and is edge-transitive but not vertex-transitive unless $r=s$, while its line graph $L(K_{r,s})$ is vertex-transitive and $(r+s-2)$--regular.

\begin{corollary}\label{cor:eta-LKmn}
Let $s\geq r \geq 1$. For the rook's graphs $K_r\square K_s$, the incompatibility robustness of the corresponding set of $rs$ observables satisfies
\begin{equation}
\eta\big(K_r\square K_s\big)\;\le\;\frac{1}{\sqrt{s}},
\end{equation}
with equality if and only if a partial Hadamard matrix of size $r\times s$ exists.
\end{corollary}

The proof is given in Appendix~\ref{A:line_examples}. The robustness can be expressed, via Theorem~\ref{thm:spectra}, as a sum involving the singular values of an $r\times s$ matrix $C\in\{\pm 1\}^{r\times s}$. The bound follows from the Cauchy-Schwarz inequality and is saturated when $CC^\top = s\,\id_r$, that is, when $C$ forms a partial Hadamard matrix~\cite{delauney10}. Such a matrix, whose rows are pairwise orthogonal, reduces to a Hadamard matrix when $r=s$.

In general, an $r\times s$ partial Hadamard matrix can exist only if $r\le s$, since no more than $s$ orthogonal vectors exist in $\mathbb{R}^s$. In addition, the number of columns $s$ must be even for $r\geq 2$, and a multiple of four for $r\geq 3$~\cite{delauney10}. It is conjectured that these conditions are also sufficient for existence.

\begin{figure}[t!]
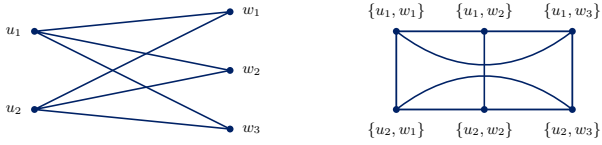

  \centering
  \rooksCurved
  \caption{The graph $K_{2,3}$ (left) and its line graph $L(K_{2,3})=K_2 \square K_3$ (right), the $2\times3$ rook’s graph.}
  \label{fig:rook}
\end{figure}

\subsubsection{Paths}\label{sec:paths}

A path graph $P_n$ consists of $n$ vertices whose edges can be arranged linearly. The line graph of a path is a path with one fewer vertex, i.e., $P_n=L(P_{n+1})$. Unlike the preceding examples, path graphs are not edge-transitive, therefore the equality in Thm. \ref{thm:spectra} no longer applies. However, we can derive upper and lower bounds on the incompatibility robustness that remain asymptotically tight.

\begin{corollary}
    \label{prop:eta_paths}
    Let $P_n$ be a path on $n$ vertices. For $n$ even:
\begin{equation*}\label{eq:path_even}
 \frac{2}{n+2}\,\csc\!\Big(\frac{\pi}{n+2}\Big) \leq \eta(P_n)\leq \frac{1}{n}\,\cot\!\Big(\frac{\pi}{2(n+2)}\Big) - \frac{1}{n}\,,
    \end{equation*}
and for $n$ odd:
  \begin{equation*}\label{eq:path_odd}
 \frac{2}{n+1}\,\csc\!\Big(\frac{\pi}{n+1}\Big) \leq \eta(P_n)\leq \frac{1}{n}\,\csc\!\Big(\frac{\pi}{2(n+2)}\Big) - \frac{1}{n}\,.
    \end{equation*}
Furthermore, $\eta(P_n)\rightarrow\frac{2}{\pi}$ as $n\rightarrow\infty$.
\end{corollary}

The upper bounds follow from the equality $\mathcal{E}_s^{\max}(P_n)=\mathcal{E}(P_n)$ \cite{shader09} applied to Thm. \ref{thm:spectra}, as well as known expressions for $\mathcal{E}(P_n)$ \cite{gutman12}. The lower bounds can be derived from Cor. \ref{prop:eta-cycle} together with the observation that $\eta(P_m)\;\geq \;\eta(C_{n})$ for all $m<n$, as shown in Appendix \ref{A:line_examples}.

The incompatibility of path graphs exhibit behaviour analogous to that of even cycles. Comparing with exact values obtained by semidefinite programming (see Fig. \ref{fig:eta-side-by-side}), we are led to the following prediction:
\begin{conjecture}
For all $n\geq 1$, paths satisfy
\[\eta(P_{2n})=\eta(P_{2n+1})=\eta(C_{2n+2})=\frac{2}{2n+2}\csc\Big(\frac{\pi}{2n+2}\Big)\,. 
\]
\end{conjecture}

\begin{figure*}[t]
  \centering
  \subfloat[Cycles\label{fig:a}]{
    \includegraphics[width=0.4\linewidth]{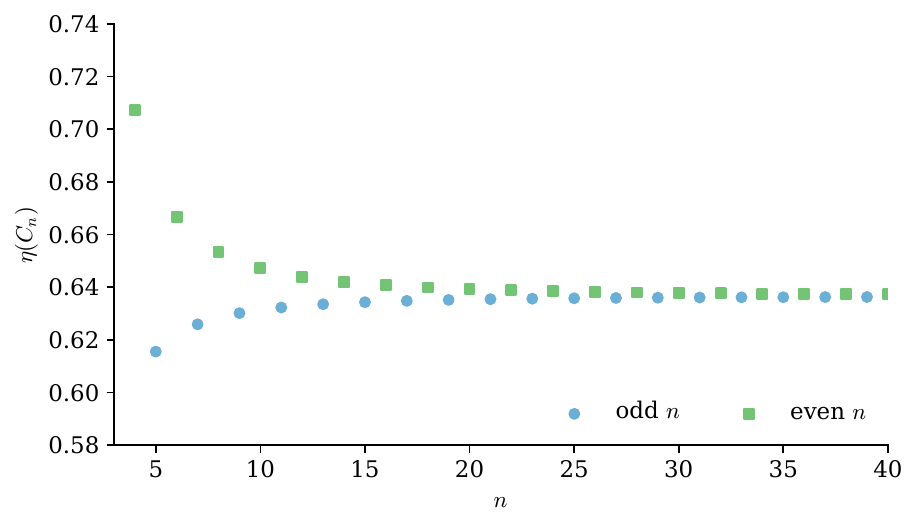}}
   \hspace{0.08\linewidth}
  \subfloat[Paths\label{fig:b}]{
    \includegraphics[width=0.4\linewidth]{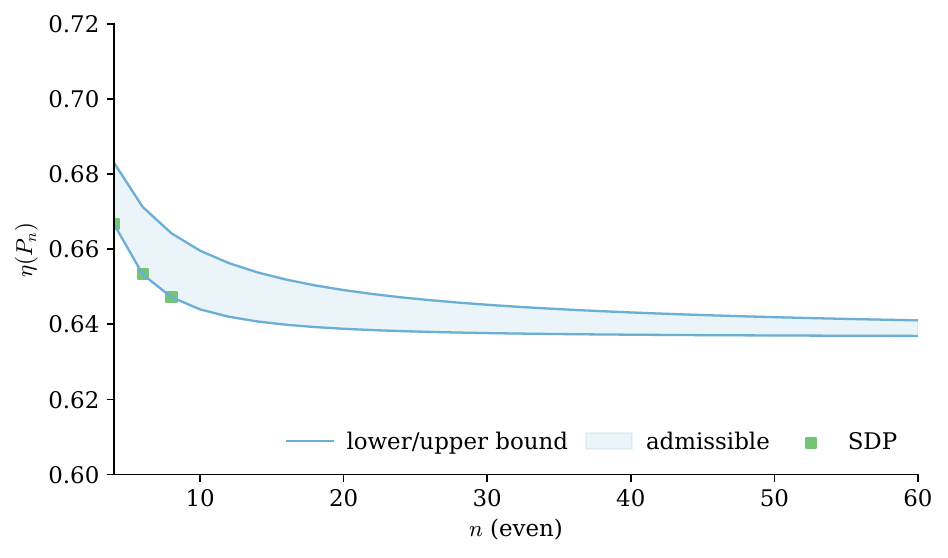}}

  \caption{Incompatibility robustness for (a) cycles and (b) paths (cf. Cor.~\ref{prop:eta-cycle} and Cor.~\ref{prop:eta_paths}). Left: exact values of $\eta(C_n)$ for odd (blue) and even (green) $n$. As $n$ increases, the incompatibility of odd (even) cycles decreases (increases), with both approaching $2/\pi$ as $n\to\infty$. Right: bounds on $\eta(P_n)$ for even $n$, together with exact values for $n\leq 9$ obtained via semidefinite programming. 
Similar bounds also hold for odd $n$, and both cases exhibit the same asymptotic behaviour as cycles.}
  \label{fig:eta-side-by-side}
\end{figure*}

\subsection{Optimally incompatible line graphs}\label{sec:max_inc}

We now consider a general upper bound on the incompatibility robustness of line graphs, and attempt to determine for which graphs this bound is saturated.

\begin{proposition}\label{prop:optimal_inc}
    Let $G=(V,E)$ be a simple graph with with $|V|=n$, $|E|=m$, and maximum
degree $\Delta$. Then its line graph satisfies
\begin{equation}\label{eq:max}
  \eta(L(G))\leq \frac{n}{2m}\sqrt{\Delta}\,.  
\end{equation}
\end{proposition}
\begin{proof}
     It was shown in \cite[Thm. 2.5]{adiga10} that the maximal skew--energy of a graph $G$ satisfies $\mathcal{E}^{\max}_s(G)\le\sqrt{2mn}\le n\sqrt{\Delta}$. This follows from $\mathcal{E}_s(G^\sigma)=\sum_j|\lambda_j(S)|\leq \sqrt{\sum_j|\lambda_j(S)|^2}\sqrt{n}=\sqrt{2mn}\leq \sqrt{(n\Delta)n}$, where $\sigma$ is any orientation of $G$. The first inequality relies on Cauchy-Schwarz, and the subsequent equality holds due to the relation $\sum_j|\lambda_j(S)|^2=2m$. The final inequality follows from $2m\leq n\Delta$, which is a consequence of $|\lambda_j(S)|\leq \Delta$ for all $j$ (see Ref. \cite{adiga10}). Since the line graph $L(G)$ has $m$ vertices, Thm. \ref{thm:spectra} implies that the incompatibility robustness satisfies  $\eta(L(G))\leq \frac{n}{2m}\sqrt{\Delta}$.
\end{proof}

\begin{definition}\label{def:optimal} We say that a line graph $L(G)$ is \emph{optimally incompatible} if the bound \eqref{eq:max} is saturated.
\end{definition}

For which line graphs is the optimum value achieved? The next result gives a necessary condition on the skew--adjacency matrix of $G$.

\begin{proposition}\label{prop:weighing}
Let $G$ be a simple graph with $n$ vertices, $m$ edges and maximum degree $\Delta$. If $L(G)$ is optimally incompatible, then there exists an orientation of $G$ such that its skew-adjacency matrix $S$ satisfies
\begin{equation}\label{eq:orthS}
   S^\top S \;=\; \Delta\,\id_n\,.
\end{equation}
Consequently, $G$ must be $\Delta$--regular.
\end{proposition}

\begin{proof}
According to \cite[Cor.~2.6]{adiga10}, the maximal skew--energy of a graph $G$ satisfies $\mathcal{E}_s^{\max}(G)=n\sqrt{\Delta}$ if and only if there exists an orientation $\sigma$ such that $S(G^\sigma)^\top S(G^\sigma)=\Delta \id_n$. To see this, recall from the proof of Prop.~\ref{prop:optimal_inc} that
$\mathcal{E}_s(G^\sigma)=\sum_j|\lambda_j(S)|\le \sqrt{n\sum_j|\lambda_j(S)|^2} =\sqrt{2mn}\le n\sqrt{\Delta}$. The first inequality, which applies Cauchy--Schwarz, holds with equality if and only if all eigenvalues of $S$ have the same modulus, i.e., $|\lambda_j(S)|=\alpha$ for all $j$. Since $S^\top S$ is real symmetric and all its eigenvalues are equal to $\alpha^2$, it follows that $S^\top S=\alpha^2 \id_n$. Furthermore, $\sum_j|\lambda_j(S)|^2=\mathrm{tr}(S^\top S)=2m$,therefore $\alpha=\sqrt{2m/n}$. The inequality $2m \le n\Delta$ (used to derive the general bound) follows from the fact that the average degree of a graph cannot exceed its maximum degree: $2m/n = \sum_{v\in V}\deg(v)/n \le \max_{v\in V}\deg(v) = \Delta$. Equality holds if and only if all vertices have degree $\Delta$, i.e., $G$ is $\Delta$--regular. When both equalities hold, $|\lambda_j(S)|=\sqrt{\Delta}$ for all $j$, which implies $S(G^\sigma)^\top S(G^\sigma)=\Delta \id_n$. Therefore, by Thm.~\ref{thm:spectra}, the bound \eqref{eq:max} is saturated only if such an orientation exists. Finally, this condition immediately implies $\Delta$--regularity, since $(S^\top S)_{vv}=\sum_u S_{uv}^2=\deg(v)=\Delta$ for every $v\in V$.
\end{proof}

The matrix condition \eqref{eq:orthS} coincides with the definition of a weighing matrix. An $n\times n$ weighing matrix $W$ with weight $k$, denoted $W(n,k)$, has entries in $\{0,\pm 1\}$ such that $W^\top W=k\,\id_n$. These matrices are closely connected to orthogonal designs \cite{colbourn10} and have a variety of applications \cite{koukouvinos97}. Prop.~\ref{prop:weighing} therefore implies that for $L(G)$ to be optimally incompatible, $G$ must admit an orientation whose skew-adjacency matrix is a skew-weighing matrix $W(n,k)$.

When $k=n-1$, weighing matrices are equivalent to conference matrices. Thus, for a $(n-1)$--regular graph $G$, optimal incompatibility requires the existence of a skew-conference matrix, as confirmed by the example $G=K_n$ in Cor. \ref{cor:johnson}. For the hypercube $Q_d$, we showed in Cor.~\ref{cor:eta-hypercube-exact} that $L(Q_d)$ attains optimal incompatibility, implying that a skew--weighing matrix $W(2^d,d)$ exists for all $d$. The complete bipartite graph $K_{r,s}$, by contrast, is regular only when $r=s$. Together with Cor.~\ref{cor:eta-LKmn}, it follows that $L(K_{r,s})$ is optimally incompatible exactly when $r=s$ and an $s\times s$ Hadamard matrix $H$ exists. In this case, the corresponding skew--weighing matrix $W(2s,s)$ can be constructed with off--diagonal blocks given by $H$ and $-H^\top$. In summary:

\begin{corollary}\label{cor:optimal}
    The line graph of the hypercube $Q_d$ is optimally incompatible for every $d$. The Johnson graph $J(n,2)$ is optimally incompatible if and only if a skew--conference matrix of order $n-1$ exists. The rook's graph $K_r\square K_s$ is optimally incompatible if and only if $r=s$ and a Hadamard matrix of order $s$ exists.
\end{corollary}

\begin{proof}
    This follows from Corollaries \ref{cor:johnson}, \ref{cor:eta-hypercube-exact} and \ref{cor:eta-LKmn}.
\end{proof}

In each of these examples, the underlying graph $G$ is $k$--regular, as required by Prop. \ref{prop:weighing}. In the line graph picture, this translates into $(2k-2)$--regularity, meaning that each quadratic observable anti-commutes with exactly $2(k-1)$ others. Thus, if optimal incompatibility is achieved, we require anti-commutativity to be distributed homogeneously across the observables. However, as $k$--regular graphs do not always admit an orientation with
$S^\top S=k \id$, regularity alone does not guarantee
optimal incompatibility. For instance, although $C_n$ is 2-regular, Cor.~\ref{prop:eta-cycle} implies
$\eta(C_n)\to 2/\pi<1/\sqrt{2}$ as $n\to\infty$, which is strictly below the value of the bound. 

A complete classification of $k$--regular graphs with an orientation that satisfies \eqref{eq:orthS} remains an open problem except for small cases. For $3$--regular graphs, the only connected examples are $K_4$ and $Q_3$
\cite{gong12}. Classifications, which contain infinite families, are also known for $4$- and $5$-regular graphs \cite{chen13,guo17}.

\section{Merged Johnson graphs}\label{sec:fermionic}

We now consider the incompatibility of degree--$k$ Majorana observables---a class of observables that are relevant in condensed-matter physics, quantum chemistry, and quantum computation \cite{babbush23,clinton24,zhao21,wan23,chapman22,mcnulty24}. We consider a system consisting of $n$ Majorana modes, corresponding to $n/2$ fermionic modes, represented by Hermitian Majorana operators $\Gamma_1,\ldots,\Gamma_n$ satisfying the relations in \eqref{eq:majoranas}.

For a fixed integer $k$, let $\mathscr{I}$ be a $k$--subset of $[n]$, and denote by $\binom{[n]}{k}$ the set of all such subsets. 
A degree--$k$ observable is defined as
\begin{equation}\label{eq:ferm_monomials}
A_\mathscr{I} := i^{\alpha} \prod_{j\in \mathscr{I}} \Gamma_j\,,
\end{equation}
where the product is taken in increasing index order. The family of all degree--$k$ observables is then
\begin{equation}\label{eq:ferm_monomials_set}
\mathcal{A}_k=\{A_\mathscr{I}: \mathscr{I}\in\binom{[n]}{k}\}\,.
\end{equation}
By a suitable choice of the phase exponent, e.g., $\alpha=\lfloor k/2\rfloor$, the observables satisfy $A_\mathscr{I}^\dagger=A_\mathscr{I}$ and $A_\mathscr{I}^2=\id$.

In the quadratic case $k=2$, the anti-commutativity graph of $\mathcal{A}_2$ coincides with the Johnson graph $J(n,2)$ treated in Cor.~\ref{cor:johnson}. For higher degrees $k>2$, the corresponding graphs are no longer line graphs, and Theorem~\ref{thm:spectra} does not apply.

For arbitrary $k$, the anti-commutation relations within $\mathcal{A}_k$ are
\begin{equation}\label{eq:car}
\{A_\mathscr{I},A_{\mathscr{I}'}\}=0 
\ \iff\ 
\begin{cases}
|\mathscr{I}\cap \mathscr{I}'|\ \text{is odd}, & k\ \text{even},\\[2pt]
|\mathscr{I}\cap \mathscr{I}'|\ \text{is even},& k\ \text{odd}.
\end{cases}
\end{equation}
Let $V=\binom{[n]}{k}$. The generalised Johnson graph $J(n,k,i)$ has vertex set $V$ with an edge between
$\mathscr I$ and $\mathscr I'$ if and only if $|\mathscr I\cap\mathscr I'|=i$. 
For $L\subseteq\{0,1,\dots,k-1\}$, the merged Johnson graph $J(n,k,L)$ is the union of $J(n,k,i)$ over $i\in L$.
Thus, the anti-commutativity graph of $\mathcal{A}_k$ is
\begin{equation}\label{eq:merged_graph}
G=
\begin{cases}
J\!\left(n,k,\{1,3,\ldots,k-1\}\right), & k\ \text{even},\\[4pt]
J\!\left(n,k,\{0,2,\ldots,k-1\}\right), & k\ \text{odd}.
\end{cases}
\end{equation}

A recent result by Linz~\cite{linz24} shows that for fixed $k$, the Lovász number of merged Johnson graphs satisfies
\begin{equation}\label{eq:theta_merged}
\vartheta\!\left(J(n,k,L)\right) = \Theta\!\big(n^{\,k-|L|}\big).
\end{equation}
Since $|V|=\binom{n}{k}=\Theta(n^k)$, Theorem~\ref{thm:lovasz} yields the upper bound
\begin{equation}
\eta\!\left(J(n,k,L)\right) = \mathcal{O}\!\big(n^{-|L|/2}\big).
\end{equation}
For the specific anti-commutativity graphs~\eqref{eq:merged_graph}, we obtain the bounds
\begin{equation}\label{eq:theta_merged_bound}
\eta(G)=
\begin{cases}
\mathcal{O}\!\big(n^{-k/4}\big), & k\ \text{even},\\[4pt]
\mathcal{O}\!\big(n^{-(k+1)/4}\big), & k\ \text{odd}.
\end{cases}
\end{equation}

The clique number of~\eqref{eq:merged_graph}, when applied to Prop.~\ref{prop:clique}, yields a considerably weaker bound than~\eqref{eq:theta_merged_bound} (see Appendix~\ref{A:lsystems}). For large $n$, the fractional chromatic number can be determined using knowledge of the automorphism group of $J(n,k,L)$, which was established in~\cite{jones05} (see Appendix~\ref{A:degreeq}). Applying Prop.~\ref{prop:frac} then yields the asymptotic lower bounds $\eta(G)=\Omega(n^{-k/2})$ when $k,n$ are even and $2\leq k\leq n/2$ (see Cor. \ref{cor:johnson_fractional} in Appendix~\ref{A:degreeq}).

For fixed $k=\mathcal{O}(1)$, this lower bound can be improved to $\eta(G)=\Omega(n^{-k/4})$ by a joint measurement constructed using fermionic Gaussian unitaries \cite{mcnulty24}. Together with the upper bound in Eq.~\eqref{eq:theta_merged_bound}, this yields tight asymptotic bounds, extending the results previously established for $k\leq 10$~\cite{mcnulty24}.

\begin{proposition}\label{prop:inc_rob_merged}
Assume $n$ and $k$ are even, with $k$ fixed and independent of $n$.  
Then the incompatibility robustness of the observables $\mathcal{A}_{k}$ with anti-commutativity graph~\eqref{eq:merged_graph} satisfies $\eta(G)=\Theta(n^{-k/4})$.
\end{proposition}

This result establishes the asymptotic scaling conjectured in~\cite[Conj.~1]{mcnulty24}, confirming that the fermionic Gaussian joint measurement constructed there is asymptotically optimal.

The incompatibility robustness of these measurements limits the efficiency at which estimation tasks such as partial fermionic tomography and energy approximations of electronic-structure Hamiltonians can be achieved with joint measurement strategies \cite{zhao21,wan23,mcnulty24,majsak24}. For arbitrary $n$ and $k$, it remains open to find the optimal joint measurements for this task.

\section{Summary and outlook}\label{sec:summary}

We have shown that for families of binary observables, measurement incompatibility is intimately linked to the structural properties of graphs encoding their anti-commutation relations. In particular, the Lovász number, clique number, fractional chromatic number, energy, and skew-energy emerge as graph invariants that bound---sometimes tightly---the incompatibility robustness of these observables. For certain graph families, these quantities yield exact or asymptotic expressions for the robustness.

We have also seen, through our analysis of line graphs, connections between joint measurability and extremal problems in combinatorics and matrix theory, involving structures such as weighing matrices, Hadamard matrices, and skew-conference matrices. It remains an open question whether, and under what conditions, these combinatorial objects arise in the explicit construction of joint measurements that saturate the incompatibility bounds. Such constructions would be useful for various estimation tasks in quantum computing (see, e.g., \cite{mcnulty22,mcnulty24,majsak24}).

Some interesting open questions arise when considering how modifications of the anti-commutation pattern between observables affect incompatibility. For instance, how does the robustness change when edges or vertices are removed? Interestingly, semidefinite programming suggests that an even path $P_n$, which has fewer anti-commuting pairs (and thus more mutual commutativity) than a cycle $C_n$, can be more incompatible (cf. Fig. \ref{fig:eta-side-by-side}). More generally, how do graph operations---such as complementation, Cartesian or tensor products, or graph joins---affect the behaviour of $\eta(G)$?

Another interesting direction is to relax the constraints on the observables represented by a graph. While this work focused on binary observables that are unitary and either commute or anti-commute, it is natural to consider what changes when these assumptions are lifted. In particular, one may consider settings where non-adjacent observables do not necessarily commute, or where the observables are not unitary. For instance, the bound in Theorem~\ref{thm:lovasz} continues to hold even without assuming that non-edges correspond to commuting pairs. In this broader setting, it is an open question when the incompatibility robustness remains invariant under different realisations of the anti-commutativity graph, and for which graph families it can be quantified analytically.

\section*{Acknowledgements}

The author acknowledges support from PNRR MUR Project No. PE0000023-NQSTI. The author thanks Saverio Pascazio, Glen Bigan Mbeng and Micha\l{} Oszmaniec for useful discussions on earlier versions of the manuscript.

\bibliographystyle{apsrev4-2}
\bibliography{ref_graphs}

\onecolumngrid

%
%
%
\appendix
\section{Graph theory preliminaries}\label{A:prelims}

Let $G=(V,E)$ be a graph with vertex set $V=\{1,\ldots,n\}$ and edge set $E\subseteq \binom{V}{2}$. We sometimes use the notation $V(G)$ and $E(G)$ to emphasize the dependence of the vertex and edge sets on $G$. In this appendix, we introduce the essential graph-theoretic definitions neccessary in this work, including relevant graph parameters, graph symmetries, and oriented graphs. A comprehensive overview of algebraic graph theory can be found in \cite{godsil01}.

\subsection*{Graph parameters}

The \emph{independence number} $\alpha(G)$ is the cardinality of the largest independent set $S\subseteq V$ in which no two vertices are adjacent. The \emph{clique number} $\omega(G)$ is the size of the largest complete subgraph (clique) of $G$, and satisfies $\omega(G)=\alpha(\overline G)$, where $\overline G$ is the complement of $G$. The \emph{chromatic number} $\chi(G)$ is the smallest number of colours that can be assigned to the vertices of $G$ such that no two adjacent vertices share the same colour. Equivalently, the chromatic number of $G$ equals the \emph{clique covering number} $\theta(\overline G)$ of its complement, where the clique covering number $\theta(G)$ is the minimum number of cliques needed to cover all vertices of $G$.    An $a:b$--colouring of $G$ assigns $b$ colours to each vertex from a palette of $a$ colours, such that the colours assigned to adjacent vertices are disjoint (see Fig. \ref{fig:cycle_5}).

\begin{definition}\label{def:frac_colouring}
The \emph{fractional chromatic number} $\chi_f(G)$ is the infimum of $a/b$ such that an $a:b$--colouring of $G$ exists, for positive integers $a,b$.
\end{definition}

The \emph{Lovász number}, first introduced as an upper bound on the Shannon capacity of a graph \cite{lovasz79}, is a semidefinite programming relaxation of the independence number. Unlike $\alpha(G)$ and $\theta(G)$ (both NP-hard to compute), it is computable in polynomial time.
\begin{definition}\label{def:lovasz}
   The \emph{Lovász number} is given by
   \begin{equation}\label{Aeq:lov}
\vartheta(G):=\max_{X\in\mathbb{R}^{n\times n}} \{\tr{X \mathbb{J}}\,|\,X\geq 0\,,\,\tr{X}=1\,,\,X_{ij}=0\,\,\, \text{for all}\,\,\{i,j\}\in E\}\,,
\end{equation}
where $\mathbb{J}$ is the matrix of all ones.
\end{definition}
\noindent The well-known sandwich theorem gives $\alpha(G)\leq \vartheta(G)\leq\theta(G)$, or equivalently, $\omega(G)\leq \vartheta(\overline G)\leq\chi(G)$ \cite{knuth93}. 

For a graph $G=(V,E)$ encoding, via Def.~\ref{def:anti_graph}, the anti-commutativity 
relations of a set of binary observables $\{A_v : v\in V\}$, let 
$a:V\to\mathbb{R}$ be an assignment of real weights to the vertices. 
The corresponding linear combination of the set of observables is
\[
H(a)=\sum_{v\in V} a_v A_v\,.
\]
In recent work by Hastings and O'Donnell \cite{hastings21}, a related sandwich-type bound, $\alpha(G)\leq\Psi(G)\leq\vartheta(G)$, was established for the quantity 
\begin{equation}\label{Aeq:max_evalue}
\Psi(G)\;:=\;\max_{\;a:V\to\mathbb{R},\;\|a\|_2=1}
\bigl\|H(a)\bigr\|_{\infty}^2\,,
\end{equation}
where $\|\cdot\|_\infty$ denotes the operator norm and $\|\cdot\|_2$ the $\ell^2$-norm. The upper bound $\Psi(G)\leq\vartheta(G)$ is relevant in this work, and we reproduce its proof in Appendix \ref{A:theorem_proof}.

\subsection*{Graph symmetries}

An automorphism of a graph $G=(V,E)$ is an isomorphism from $G$ to itself, i.e., a permutation of the vertices that preserves adjacency and non-adjacency relations. The set of all automorphisms of $G$ forms a group known as the \emph{automorphism group}, denoted by $\text{Aut}(G)$ (see, e.g., \cite{godsil01}). This group is a subgroup of the symmetric group $\text{Sym}(V)$. For example, the complete graph $K_n$ has $\text{Aut}(K_n)\cong \text{Sym}([n])$, as any permutation of its vertices preserves its structure.

Two  symmetries relevant in this work are vertex- and edge-transitivity. These indicate that the automorphism group acts transitively on the vertices or edges of the graph, so that any two can be mapped to each other by an automorphism.

\begin{definition}\label{def:transitive}
Let $G=(V,E)$ be a graph and $\mathrm{Aut}(G)$ its automorphism group. The graph is \emph{vertex-transitive} if $\mathrm{Aut}(G)$ acts transitively on $V$, 
i.e., for any two distinct $v,v'\in V$ there exists $f\in\mathrm{Aut}(G)$ with $f(v)=v'$. 
It is \emph{edge-transitive} if $\mathrm{Aut}(G)$ acts transitively on $E$, 
i.e., for any $e,e'\in E$ there exists $f\in\mathrm{Aut}(G)$ with $f(e)=e'$.
\end{definition}

For example, cycles $C_n$, complete graphs $K_n$, hypercubes $Q_d$, and (generalised) Johnson graphs are both vertex- and edge-transitive. Complete bipartite graphs $K_{r,s}$ are edge-transitive for all $r,s$, but vertex-transitive only when $r=s$. Rook’s graphs $K_r\square K_s$ are vertex-transitive for all $r,s$, and edge-transitive when $r=s$. Paths $P_n$ are neither vertex- nor edge-transitive for $n\geq 4$.

\subsection*{Oriented graphs}

An \emph{oriented graph} $G^\sigma$ is obtained from a simple undirected graph $G=(V,E)$ by assigning a direction to each edge, 
so that every $\{u,v\}\in E$ is replaced by either $u\to v$ or $v\to u$.

\begin{definition}\label{def:oriented}  
An \emph{oriented graph} $G^\sigma$ of an underlying simple graph $G=(V,E)$ is specified by a function
\[
\sigma:\{(u,v),(v,u):\{u,v\}\in E\}\to\{\pm1\},
\]
satisfying $\sigma(u,v)=-\sigma(v,u)$ for all $\{u,v\}\in E$.  
We interpret $\sigma(u,v)=1$ as the edge $\{u,v\}$ being directed $u\to v$.    
\end{definition}
The \emph{adjacency matrix} $A\in\mathbb{R}^{|V|\times|V|}$ of a graph $G$ and the \emph{skew-adjacency matrix} $S(G^\sigma)\in\mathbb{R}^{|V|\times|V|}$ of $G^\sigma$ are defined in the main text (cf. \eqref{eq:skew-adjacency}). The eigenvalues of $S(G^\sigma)$ are invariant under the process of \emph{switching} \cite{zaslavsky13}, which can reverse the direction of all edges incident to a vertex.

\begin{definition}\label{def:switching}
Let $G^\sigma$ be an oriented graph. Given a switching function 
$\theta:V\to\{\pm1\}$, define a new orientation $\sigma_\theta$ by 
\[
\sigma_\theta(u,v)=\theta(u)\,\sigma(u,v)\,\theta(v)\,,
\]
for each $\{u,v\}\in E$. We call $G^\sigma$ and $G^{\sigma_\theta}$ \emph{switching equivalent}, and the equivalence classes under this relation are the \emph{switching classes}.
\end{definition}

If $D(\theta)=\mathrm{diag}(\theta(v_1),\ldots,\theta(v_{|V|}))$, then 
$S(G^{\sigma_\theta})=D(\theta)^{-1}\,S(G^\sigma)\,D(\theta)$, therefore switching preserves the spectrum of $S(G^\sigma)$ \cite{zaslavsky13}.

\section{Observables from anti-commutativity graphs}\label{A:arbitrary_graphs}

We now prove that it is always possible to construct observables that satisfy the anti-commutation relations defined by an arbitrary simple graph $G$. We will assume that the graph is \emph{connected}. If a graph is disconnected, its connected components---such as isolated vertices---share no incompatibility and can be treated independently.

\begin{proposition*}[Restatement of Prop.~\ref{prop:arb_graphs}]\label{prop:arb_graphs_restated}
For every simple connected graph $G = (V,E)$, there exists a set of binary observables whose anti-commutativity graph is $G$.  
The construction~\eqref{eq:clifford_monomial_main} uses at most $|E| + |V|$ Majorana operators, therefore the observables act on a Hilbert space of dimension at most $2^{\lfloor \frac{|V| + |E|}{2} \rfloor}$.
\end{proposition*}

\begin{proof}
Let $G = (V,E)$ be a simple graph.  Assign a distinct label from $1,\dots,|E|$ to each edge. For each vertex $v\in V$, let $\mathscr{I}(v)$ denote the set of integer labels corresponding to edges incident to $v$. If $\deg(v)$ is odd, we append a unique auxiliary index $d_v \in \{|E|+1, \dots, |V|+|E|\}$ to $\mathscr{I}(v)$, disjoint from all edge labels, which ensures $\mathscr{I}(v)\subseteq \{1,\ldots,|E|\}\cup \{d_v\}$ has even cardinality. 

Define the observable
\begin{equation}\label{eq:clifford_monomial}
A_{\mathscr{I}(v)}:=i^{|\mathscr{I}(v)|/2}\prod_{j\in \mathscr{I}(v)}\!\!\Gamma_j\,,
\end{equation}
where $\Gamma_j$ are the Majorana operators in \eqref{eq:majoranas}, and the product is taken in increasing order of indices. The phase ensures that $A_{\mathscr{I}(v)}$ is Hermitian and satisfies $A_{\mathscr{I}(v)}^2 = \id$.

The anti-commutation relations any between degree--$k$ Majorana monomials $\mathscr{I}$ and $\mathscr{I}'$ are given in Eq. \eqref{eq:car}. In particular, if $|\mathscr{I}|$ and $|\mathscr{I}'|$ are even, then $A_\mathscr{I}$ and $A_{\mathscr{I}'}$ anti-commute if and only if $|\mathscr{I} \cap \mathscr{I}'|$ is odd. In the above construction we have $|\mathscr{I}(v)\cap \mathscr{I}(v')|=1$ if and only if $\{v,v'\}\in E$, and zero otherwise. It follows that $A_{\mathscr{I}(v)}$ and $A_{\mathscr{I}(v')}$ anti-commute if and only if $\{v,v'\} \in E$. Hence, the anti-commutativity graph of the set $\mathcal{A} = \{A_{\mathscr{I}(v)} \mid v \in V\}$ is $G$.

The total number of Majorana operators used is $|E|$ (one for each edge) plus at most $|V|$ (one auxiliary index per vertex if needed). Moreover, the construction can be implemented efficiently, using only the edge set and degree information of the graph.
\end{proof}

To illustrate the above construction, consider the path graph $P_4$ on four vertices, labelled $V=\{v_1,v_2,v_3,v_4\}$. It has three edges and two vertices of odd degree. We assign the labels $1$, $2$, $3$ to the edges $\{v_1,v_2\}$, $\{v_2,v_3\}$, $\{v_3,v_4\}$, respectively. The sets labelling the edges incident to each vertex are
\begin{equation}
\mathscr{I}(v_1) = \{1, 4\}\,, \quad \mathscr{I}(v_2) = \{1, 2\}\,, \quad \mathscr{I}(v_3) = \{2, 3\}\,, \quad \mathscr{I}(v_4) = \{3, 5\}\,,
\end{equation}
where we have added auxiliary labels $4$ and $5$ to $\mathscr{I}(v_1)$ and $\mathscr{I}(v_4)$, respectively, to ensure even cardinality. These define the observables 
\begin{equation}
A_{\mathscr{I}(v_1)}=i\Gamma_1\Gamma_4\,, \quad A_{\mathscr{I}(v_2)} = i\Gamma_1\Gamma_2\,, \quad A_{\mathscr{I}(v_3)}= i\Gamma_2\Gamma_3\,, \quad A_{\mathscr{I}(v_4)} = i\Gamma_3\Gamma_5\,,
\end{equation}
which are easily verified to have the anti-commutativity graph $G=P_4$.

\begin{remark}
The above construction, which at worst requires a Hilbert space of dimension $2^{\lfloor \frac{|V|+|E|}{2}\rfloor}$, may significantly exceed the minimal dimension $2^{r/2}$ required to realise a set of observables whose anti-commutativity graph is $G$, where $r$ denotes the rank of the adjacency matrix of $G$ over $\mathbb{F}_2$.
\end{remark}

To justify this remark, we associate with each graph $G$ a complex algebra $\mathfrak{C}(G)$ as described in the proof of Prop. \ref{prop:eta_invariance}. For an arbitrary graph $G$, the algebra $\mathfrak{C}(G)$ has complex vector space dimension $2^{|V|}$, and a basis is given by all monomials in the generators $x_v$, for $v\in V$. Since $\mathfrak{C}(G)$ is semisimple (see, e.g., \cite{khovanova08}), the Artin-Wedderburn theorem \cite{lam91} implies that it is isomorphic to a direct sum of $k$ copies of a matrix algebra $M_{d}(\mathbb C)$:
\begin{equation}\label{eq:isomorphism}
\mathfrak C(G)\;\cong\;\bigoplus_{i=1}^k M_{d}(\mathbb C)\,.
\end{equation}
Let $B$ denote the adjacency matrix of $G$, and let $r=\rank_{\mathbb F_2}(B)$ be the rank over the field $\mathbb{F}_2$. Since $B$ is symmetric with zero diagonal, it defines an alternating bilinear form over $\mathbb{F}_2$, and thus $r$ is necessarily even. It is known (see, e.g., \cite{khovanova08}) that the centre of $\mathfrak{C}(G)$ has dimension $2^{|V|-r}$. This fixes the number of simple components in the decomposition \eqref{eq:isomorphism} to $k=2^{|V|-r}$. Thus, by comparing dimensions, we obtain $2^{|V|} = k \cdot d^2$, which implies $d = 2^{r/2}$.

It follows that an irreducible representation of $\mathfrak{C}(G)$ has dimension $d = 2^{r/2}$, corresponding to the size of each simple matrix block in \eqref{eq:isomorphism}.
This provides a lower bound on the Hilbert space dimension required to realise the full set of (anti-)commutation relations encoded by the graph $G$, in agreement with the above remark. 

As an example, consider the complete graph $G = K_n$, whose graph algebra $\mathfrak{C}(K_n)$ coincides with the standard Clifford algebra $\mathrm{Cl}_n(\mathbb{C})$. The irreducible representations of this algebra have dimension $2^{\lfloor n/2 \rfloor}$. In particular, when $n$ is even, the adjacency matrix has rank $r = n$, and the algebra is simple, i.e., $\mathfrak{C}(K_n) \cong M_{2^{n/2}}(\mathbb{C})$. When $n$ is odd, we have $r = n - 1$, and $\mathfrak{C}(K_n)$ decomposes as a direct sum of two simple components, each isomorphic to $M_{2^{(n-1)/2}}(\mathbb{C})$.

\section{Proof of Theorem 1}\label{A:theorem_proof}

\begin{restate}[Restatement of Thm.~\ref{thm:lovasz}]\label{thm:lovasz_restated}
Let $G=(V,E)$ be the anti-commutativity graph of $\mathcal{A}=\{A_v : v\in V\}$, and let $\vartheta(G)$ be its Lovász number. Then
\begin{equation}\label{eq:lovasz_inc_bound}
\eta(G) \;\le\; \sqrt{\vartheta(G)/|V|}\,.
\end{equation}
\end{restate}

To prove Thm \ref{thm:lovasz}, we require the following two results. The first (Prop. \ref{Aprop:inc_rob_bound}) generalises a bound on the incompatibility robustness for generators of the Clifford algebra (i.e., pairwise anti-commuting operators), previously derived in \cite{kunjwal14}. We adapt the argument from \cite{kunjwal14,liang11,kunjwal13} to cover observables $\mathcal{A}=\{A_v\,|\,v\in V\}$ defined in Eq. \eqref{eq:binary_observables}. The second result (Prop. \ref{lem:hastings}) establishes $\Psi(G)\leq\vartheta(G)$, where $\vartheta(G)$ and $\Psi(G)$ are given in \eqref{Aeq:lov} and \eqref{Aeq:max_evalue}, respectively. The inequality was proved in \cite{hastings21} but we include the argument here for completeness.

\begin{proposition}\label{Aprop:inc_rob_bound}
Let $\mathcal{A}=\{A_v\,|\,v\in V\}$ be a set of binary observables, as defined in Eq. \eqref{eq:binary_observables}, with anti-commutativity graph $G$, and let $a:V\rightarrow \{\pm 1\}$ be function assigning the sign $a_v\in\{\pm 1\}$ to each vertex. Then
\begin{equation}
\eta(G)\leq\frac{1}{|V|}\max_{a:V\rightarrow \{\pm 1\}} \Big\|\sum_{v\in V}a_vA_v\Big\|_{\infty}\,,
\end{equation}
where $\no{\cdot}_{\infty}$ denotes the operator norm.
\end{proposition}
\begin{proof}
Let $V=\{1,\ldots,n\}$ and let $\{\M_v^\eta\,|\,v\in V\}$ be a collection of binary POVMs
\begin{equation}
\M^{\eta}_v(a_v)=\tfrac{1}{2}(\id+\eta a_vA_v)\,,
\end{equation}
with outcomes $a_v\in\{\pm 1\}$. Since $A_v^2=\id$ for all $v\in V$, we have:
\begin{eqnarray}
\eta&=&\frac{1}{d}\sum_{a_v\in\{\pm 1\}}\tr{a_vA_v\M_v(a_v)}\\
&=&\frac{1}{d|V|}\sum_{v\in V}\sum_{a_v}\tr{a_vA_v\M_v(a_v)}\,.
\end{eqnarray}
Now, suppose $\{\M^{\eta}_v\,|\,v\in V\}$ are jointly measurable with parent POVM $\E$ such that
\begin{equation}
\M^{\eta}_v(a_v)=\sum_{a_1,\ldots,a_{v-1},a_{v+1},\ldots,a_{n}}\E(a_1,\ldots,a_{n})\,.
\end{equation}
Substituting this into the expression for $\eta$, we get:
\begin{eqnarray}
\eta&=&\frac{1}{d|V|}\sum_{v\in V}\sum_{a_v}\text{tr}\Big[a_vA_v\sum_{a_1,\ldots,a_{v-1},a_{v+1},\ldots,a_{n}} \E(a_1,\ldots,a_{n})\Big]\\
&=&\frac{1}{d|V|}\sum_{v\in V}\sum_{a_1,\ldots,a_n} \text{tr}\Big[a_vA_v\E(a_1,\ldots,a_n)\Big]\\
&=&\frac{1}{d|V|}\sum_{a_1,\ldots,a_{n}}\text{tr}\Big[\Big(\sum_{v\in V} a_vA_v\Big)\E(a_1,\ldots,a_{n})\Big]\\
&\leq &\frac{1}{d|V|}\sum_{a_1,\ldots,a_{n}}\Big\|\sum_{v\in V}a_vA_v\Big\|_{\infty}\tr{\E(a_1,\ldots,a_{n})}\\
&\leq&\frac{1}{d|V|}\max_{a_1,\ldots,a_{n}}\Big\|\sum_{v\in V}a_vA_v\Big\|_{\infty}\sum_{a_1,\ldots,a_{n}}\tr{\E(a_1,\ldots,a_{n})}\\
&=& \frac{1}{|V|}\max_{a_1,\ldots,a_{n}}\Big\|\sum_{v\in V}a_vA_v\Big\|_{\infty}
\end{eqnarray}
where the first inequality follows from $\tr{AB} \leq \no{A}_{\infty}\tr{B}$, with $A$ Hermitian and $B$ positive semidefinite, and the second inequality follows by taking the maximum norm over all choices of $a_v$. Finally, from the condition $\sum_{a_1\ldots,a_{n}} \E(a_1, \ldots,a_{n})=\id$ we obtain the required bound.
\end{proof}

Before reproducing the next result, we introduce some notation. Let $\mathfrak{C}(G)$ denote the resulting algebra of the anti-commutativity graph, whose generators are the operators associated with the vertices. For example if $G$ is the complete graph, then $\mathfrak{C}(G)$ is the standard Clifford algebra. An element $H \in \mathfrak{C}(G)$ is called a sum-of-squares (SOS) element of degree--$2k$ if it can be written as $H=\sum_iL_i^\dagger L_i$, where $L_i\in\mathfrak{C}(G)$ has degree at most $k$.

\begin{definition}
A degree--$2k$ pseudoexpectation $\widetilde{\mathbb{E}}$ of $\mathfrak{C}(G)$ is a linear functional $\widetilde{\mathbb{E}}[\cdot]:\mathfrak{C}(G)\mapsto \mathbb{C}$, with $\widetilde{\mathbb{E}}[\id]=1$ and $\widetilde{\mathbb{E}}[H]\geq 0$ for all degree--$2k$ SOS $H$. We can write  $\widetilde{\mathbb{E}}[H]=\tr{X H}$ for a Hermitian $X\in\mathfrak{C}(G)$.
\end{definition}

\begin{proposition}[Hastings and O'Donnell \cite{hastings21}]\label{lem:hastings}
Let $G$ be the anti-commutativity graph of the set of observables $\mathcal{A}=\{A_v\,|\,v\in V\}$, then
\begin{equation}
\Psi(G)\leq\vartheta(G)\,.
\end{equation}
\end{proposition}

\begin{proof}
Let $|V|=n$, and let $\widetilde{\mathbb{E}}$ be any degree-2 pseudoexpectation on $\mathfrak{C}(G)$, which can be treated as a positive semidefinite (PSD) matrix in $\mathbb{C}^{n \times n}$. Consider an arbitrary unit-norm vector $\bm{a} = (a_1, \dots, a_{n}) \in \mathbb{R}^n$ and define the observable $H = \sum_{v\in V} a_v A_v$. Using the pseudoexpectation, define the matrix $X$ with entries
\begin{equation}
X_{vv'} = a_v a_{v'} \widetilde{\mathbb{E}}[A_v A_{v'}].
\end{equation}
Since $\widetilde{\mathbb{E}}$ is a pseudoexpectation, it satisfies the positivity condition $\widetilde{\mathbb{E}}[H^{\dagger}H] \geq 0$, which ensures that $X$ is PSD. From the properties of $\widetilde{\mathbb{E}}$, the matrix $X$ satisfies $\tr{X} = 1$ and $\tr{X \mathbb{J}} = \widetilde{\mathbb{E}}[H^{\dagger} H]$. Moreover, if $\{v,v'\} \in E$, then the anti-commutation relation $A_v A_{v'} = -A_{v'} A_v$ implies that $X_{vv'}$ is purely imaginary. 

To compare with $\vartheta(G)$, define $\hat{X} := \operatorname{Re}(X)$, which is PSD and satisfies $\text{tr}[\hat{X}] = 1$. If $\{v,v'\} \in E$, then $X_{vv'}$ is purely imaginary, therefore $\hat{X}_{vv'} = 0$. Thus, $\hat{X}$ satisfies the constraints in the definition of $\vartheta(G)$. Since $\tr{X \mathbb{J}} = \widetilde{\mathbb{E}}[H^{\dagger} H]$, we obtain $\widetilde{\mathbb{E}}[H^{\dagger} H] \leq \vartheta(G)$, therefore $\Psi(G) = \max_{\bm{a}} \no{H}_{\infty}^2 \leq \vartheta(G)$, completing the proof.
\end{proof}

\section{Measurement symmetries of line graphs}\label{A:sym}

In this appendix we show that the measurements associated with line graphs $L(G)$ satisfy additional symmetry properties when $G$ is edge-transitive (cf. Lemma~\ref{lem:uniform_rigid}). These symmetries are used to prove the lower bound in Theorem~\ref{thm:spectra}, as shown in the subsequent appendix. Throughout this section we adopt the assemblage/bundle notation of~\cite{nguyen20}.

\subsection*{Uniform and rigid measurements}

Consider a family of measurements indexed by a set $M$. Let $\Omega$ denote the set of all outcomes of all measurements, and define the map $\pi:\Omega \to M$ which sends each outcome to the
measurement it belongs. The tuple $(\Omega,\pi,M)$ is called a bundle \cite{nguyen20}, and the outcome set of
measurement $x\in M$ is known as a fibre over $x$, written as $\pi^{-1}(x)=\{ z\in\Omega : \pi(z)=x \}$. Outcomes of a measurement $x$ can be written as $z=(a| x)$, where $a$ is the outcome value. A \emph{measurement assemblage} on a Hilbert space $\mathcal{H}$
is a family of positive operators $\{E_z\}_{z\in\Omega} \subset \mathcal{B}(\mathcal{H})$
such that $\sum_{z\in\pi^{-1}(x)} E_z = \id$ for all $x\in M$. Thus each fibre $\pi^{-1}(x)$ defines a POVM.

A \emph{symmetry} of the assemblage is a group $\mathcal{G}$ acting on the bundle
$(\Omega,\pi,M)$ and on $\mathcal{H}$ in the following way \cite{nguyen20}.  First, each
$g\in\mathcal{G}$ induces a permutation $z\mapsto g\cdot z$ such that
$\pi(g\cdot z)=g\cdot\pi(z)$.  Second, there is a unitary representation
$U:\mathcal{G}\to\mathrm{U}(\mathcal{H})$, $g\mapsto U_g$, satisfying
\begin{equation}
E_z
=U_g\,E_{g^{-1}\cdot z}\,U_g^\dagger
\qquad\forall\,z\in\Omega,\ g\in\mathcal{G}.
\end{equation}

\begin{definition}
An assemblage is \emph{uniform} if $\mathcal{G}$ acts transitively
on $\Omega$, i.e., for any $z_1,z_2\in\Omega$ there exists
$g\in\mathcal{G}$ such that $g\cdot z_1 = z_2$.
\end{definition}

For $z\in\Omega$, the \emph{stabiliser} of $z$ is
$\mathcal{G}_z=\{g\in\mathcal{G}:g\cdot z=z\}$.  For each $g\in\mathcal{G}_z$ we have
$E_z=U_gE_zU_g^\dagger$, therefore $E_z$ commutes with $U(\mathcal{G}_z)$.

\begin{definition}
The assemblage is \emph{rigid} if the commutant satisfies
\begin{equation}
\mathrm{Comm}(U(\mathcal{G}_z))
:=\{X\in\mathcal{B}(\mathcal{H}):U_g X U_g^\dagger=X\ \forall g\in\mathcal{G}_z\}=\mathrm{span}\{E_z,\ \id-E_z\}\,.
\end{equation}
Equivalently, the representation $U$ restricted to $\mathcal{G}_z$ decomposes into two irreducible subrepresentations.
\end{definition}

\subsection*{Uniform and rigid line graph measurements}

Let $G=(V,E)$ be a simple graph. Assign one Majorana operator $\Gamma_v$ to each
$v\in V$, acting on a Hilbert space $\mathcal{H}$ and satisfying the relations
\eqref{eq:majoranas}. For each edge $e=\{u,v\}\in E$ define the quadratic observable
$A_e := i\,\Gamma_u\Gamma_v$, and the corresponding binary effects
\begin{equation}\label{Aeq:effects}
E_{(+\mid e)}= \tfrac12(\id + A_e),\qquad
E_{(-\mid e)} = \tfrac12(\id - A_e).
\end{equation}
The observables $A_e$ and $A_f$ anti-commute if and only if $e$ and $f$ share exactly
one vertex, therefore the anti-commutativity graph of $\{A_e\,|\,e\in E\}$ is the line graph
$L(G)$.

We now regard $\{E_{(\pm| e)}\}_{e\in E}$ as a measurement assemblage in the sense
of~\cite{nguyen20}. The outcome set is
\[
\Omega \;=\; \{(+| e),(-| e) : e\in E\},
\]
and the bundle map $\pi:\Omega\to E$ is defined by $\pi(\pm| e)=e$, so that each
fibre $\pi^{-1}(e)=\{(+| e),(-| e)\}$ specifies the two-outcome POVM associated
with the observable $A_e$.

Any graph automorphism $\tau\in\mathrm{Aut}(G)$ induces a permutation matrix
$P_\tau\in O(n)$ acting on the Majorana vector
$\Gamma=(\Gamma_1,\ldots,\Gamma_n)^\top$ by
$\Gamma' = P_\tau \Gamma$ such that $\Gamma'_j = \Gamma_{\tau(j)}$.
A standard result in fermionic linear optics states that every real orthogonal
transformation of Majorana operators is implemented by a fermionic Gaussian unitary \cite{bravyi02}. Hence there exists a Gaussian unitary
$U_\tau$ such that
\begin{equation}
U_\tau\,\Gamma_j\,U_\tau^\dagger = \Gamma_{\tau(j)}\,,
\end{equation}
for every $j\in [n]$. For the quadratic observables $A_e = i\Gamma_u\Gamma_v$ this implies
\begin{equation}\label{eq:group_action}
U_\tau\,A_e\,U_\tau^\dagger = \pm A_{\tau(e)},
\end{equation}
where the sign arises if a reordering of the Majoranas in $A_{\tau(e)}$ is needed.

Let $T_v(X) := \Gamma_vX\Gamma_v^\dagger$ denote local Majorana flips, which act on the
quadratic observables as
\begin{equation}\label{eq:flip}
T_v(A_e)=
\begin{cases}
-\,A_e,& v\in e,\\[2pt]
\ \,A_e,& v\notin e.
\end{cases}
\end{equation}
The flips $\{T_v\}_{v\in V}$ generate an elementary abelian group
$\mathcal{F} := \langle\, T_v : v\in V \,\rangle \cong (\mathbb{Z}_2)^{|V|}$, and
$\mathrm{Aut}(G)$ acts on $\mathcal{F}$ by permuting the generators.  
Thus, the full symmetry group is the semidirect product
\begin{equation}
\mathcal{G} \;\cong\; \mathrm{Aut}(G) \ltimes (\mathbb{Z}_2)^{|V|},
\end{equation}
with unitary representation
\begin{equation}\label{eq:group_urep}
U(\mathcal{G})
= \langle\, U_\tau,\ T_v : \tau\in\mathrm{Aut}(G),\ v\in V \,\rangle.
\end{equation}
The sign in~\eqref{eq:group_action} can always be absorbed by
composing $U_\tau$ with suitable flips $T_v$, and therefore do not affect the induced action on the measurement assemblage.

\begin{lemma}\label{lem:uniform_rigid}
If $G$ is edge-transitive, then the assemblage
$\{E_{(\pm\mid e)}: e\in E(G)\}$ (whose anti-commutativity graph is $L(G)$) is uniform and rigid under the action of $U(\mathcal{G})$ defined in \eqref{eq:group_urep}.
\end{lemma}

\begin{proof}
We first show the assemblage is uniform. Since $G$ is edge-transitive, $\mathrm{Aut}(G)$
acts transitively on $E$. For each $\tau\in\mathrm{Aut}(G)$, the unitary $U_\tau$ implements
$U_\tau\,\Gamma_v\,U_\tau^\dagger=\pm\Gamma_{\tau(v)}$, hence
$U_\tau\,A_{\{u,v\}}\,U_\tau^\dagger=\pm A_{\tau(e)}$ by \eqref{eq:group_action}.
If the sign is negative, a flip $T_w$ for $w\in e$ can be applied to obtain
$A_{\tau(e)}$ (cf.\ \eqref{eq:flip}). Thus, any edge $e$ can be mapped to any
other edge $f$ by an element of $U(\mathcal{G})$. For a fixed edge $e=\{u,v\}$ the flips
$T_u,T_v$ satisfy $T_u(A_e)=-A_e=T_v(A_e)$ and therefore exchange the effects
$E_{(+|e)}\leftrightarrow E_{(-|e)}$, as defined in \eqref{Aeq:effects}. Combining these actions, any outcome
$(\pm|e)$ can be mapped to any other $(\pm|f)$, proving uniformity.

To show rigid symmetry, fix $z=(+|e)$ with $e=\{u,v\}$, and let $\mathcal{G}_z$ be
its stabiliser. Consider the subgroup
\begin{equation}
H \;:=\; \langle\, T_w,\ T_uT_v \,:\, w\in V\setminus\{u,v\} \rangle \;\subseteq\; U(\mathcal{G}_z),
\end{equation}
with $T_w$ defined in \eqref{eq:flip} and $T_uT_v=\text{Ad}(\Gamma_u\Gamma_v)=\text{Ad}(-iA_e)$.  Thus $H$ stabilises $z$ and $\mathrm{Comm}(U(\mathcal{G}_z)) \subseteq\
\{ X : T X T^\dagger = X \text{ for all } T\in H \}$.

Consider an arbitrary operator $X=\sum_{\mathscr{I}\subseteq V}c_\mathscr{I}\,\Gamma_\mathscr{I}$, with
$\Gamma_\mathscr{I}:=\prod_{x\in \mathscr{I}}\Gamma_x$. For $w\in V\setminus\{u,v\}$, the relation
$T_w X T_w = X$ is equivalent to $[X,\Gamma_w]=0$. For a monomial $\Gamma_\mathscr{I}$,
\begin{equation}\label{eq:comm-criterion}
[\Gamma_{\mathscr{I}},\Gamma_w]=0
\quad\Longleftrightarrow\quad
\begin{cases}
w\notin \mathscr{I} \ \text{and } |\mathscr{I}| \text{ is even},\\
w\in \mathscr{I}    \ \text{and } |\mathscr{I}| \text{ is odd}.
\end{cases}
\end{equation}

Fix $\mathscr{I}$ with $c_{\mathscr{I}}\neq 0$ and impose \eqref{eq:comm-criterion} for
\emph{all} $w\in R:=V\setminus \{u,v\}$. Since $|\mathscr{I}|$ is a fixed integer, there cannot exist
$w,w'\in R$ with $w\in\mathscr{I}$ and $w'\notin\mathscr{I}$. Hence, one of the following two mutually exclusive cases occurs:

\smallskip
\noindent\emph{Case A:} $w\notin\mathscr{I}$ for all $w\in R$. Then $\mathscr{I}\subseteq\{u,v\}$.
By \eqref{eq:comm-criterion} (with any $w\in R$), $|\mathscr{I}|$ must be even. Therefore
$\mathscr{I}=\varnothing$ or $\mathscr{I}=\{u,v\}$, giving the monomials
$\id$ and $i\Gamma_u\Gamma_v$.

\smallskip
\noindent\emph{Case B:} $w\in\mathscr{I}$ for all $w\in R$. Then $R\subseteq\mathscr{I}$, with $\mathscr{I}=R\cup S$ and $S\subseteq\{u,v\}$. Now $|\mathscr{I}|=|R|+|S|$, and
\eqref{eq:comm-criterion} (with any $w\in R$) requires $|\mathscr{I}|$ to be odd. There are two subcases:

\begin{itemize}
\item[(i)] If $|V|$ is even, then $|R|=|V|-2$ is even, therefore $|S|$ is odd. Hence $|S|=1$ and
$S=\{u\}$ or $S=\{v\}$, yielding the monomials
$\Gamma_u\Gamma_{R}$ and $\Gamma_v\Gamma_{R}$, where $\Gamma_R=\prod_{x\in R}\Gamma_x$.

\item[(ii)] If $|V|$ is odd, then $|R|$ is odd, therefore $|S|$ is even. Hence $|S|\in\{0,2\}$,
giving the monomials $\Gamma_{R}$ (for $S=\varnothing$) and
$i\Gamma_u\Gamma_v\Gamma_{R}$ (for $S=\{u,v\}$).
\end{itemize}
Thus, by combining both cases, we have
\begin{equation}
X\in
\begin{cases}
\mathrm{span}\{\id,\ i\Gamma_u\Gamma_v,\ \Gamma_u\Gamma_{R},\ \Gamma_v\Gamma_{R}\}, & |V|\ \text{even},\\[2pt]
\mathrm{span}\{\id,\ i\Gamma_u\Gamma_v,\ \Gamma_{R},\ \Gamma_u\Gamma_v\Gamma_{R}\}, & |V|\ \text{odd}.
\end{cases}
\end{equation}

For $|V|$ even, $T_uT_v=\operatorname{Ad}(\Gamma_u\Gamma_v)$ flips the sign of both
$\Gamma_u\Gamma_{R}$ and $\Gamma_v\Gamma_{R}$, therefore invariance
forces their coefficients to vanish, leaving a two-dimensional span. 

For $|V|$ odd, we consider a fixed irreducible (parity) sector. Let $W:=\prod_{x\in V}\Gamma_x$.
Then $W$ commutes with each $\Gamma_k$, therefore $W=\omega\,\id$ with $|\omega|=1$ on that sector.
Writing $W=\varepsilon\,\Gamma_u\Gamma_v\,\Gamma_{R}$ with $\varepsilon\in\{\pm1\}$ (from the chosen ordering), we obtain $\Gamma_{R}=\varepsilon\,\omega\,(\Gamma_u\Gamma_v)^{-1}
= -\,\varepsilon\,\omega\,\Gamma_u\Gamma_v$ and 
$\Gamma_u\Gamma_v\Gamma_{R}
= \varepsilon\,\omega\,\id$. Thus, in the odd case the span
$\mathrm{span}\{\id,\,i\Gamma_u\Gamma_v,\,\Gamma_{R},\,\Gamma_u\Gamma_v\Gamma_{R}\}$
reduces to $\mathrm{span}\{\id,\,i\Gamma_u\Gamma_v\}$. In either case only $\id$ and
$i\Gamma_u\Gamma_v=A_e$ remain, therefore
$\mathrm{Comm}(U(\mathcal{G}_z))=\mathrm{span}\{\id,A_e\}$, and the assemblage is rigid.

\end{proof}

\section{Proof of Theorem 2}\label{sec:A:theorem_proof2}

\begin{restate}[Restatement of Thm.~\ref{thm:spectra}]
\label{thm:spectra_restated}
Let $L(G)$ be the line graph of a graph $G=(V,E)$ with maximum skew-energy $\mathcal{E}_s^{\max}(G)$. Then
\begin{equation}\label{Aeq:skew_bound_eigs}
   \eta(L(G))\;\le\; \frac{1}{2|E|}\mathcal{E}_s^{\max}(G)\,.
\end{equation}
If $G$ is edge-transitive then equality holds.
\end{restate}

\begin{proof}
Let $G=(V,E)$ be a coupling graph with $V=\{1,2,\ldots,n\}$, and let $\{\Gamma_v\}_{v\in V}$ be a set of Majorana operators. The quadratic observables generated by $G$ are $\{\,i\Gamma_u\Gamma_v\,|\,\{u,v\}\in E\}$ with $L(G)$ as their anti-commutativity graph. Let $a:E\rightarrow \{\pm 1\}$ be a sign function assigning $a_{uv}\in\{\pm 1\}$ to each edge $\{u,v\}\in E$. Equivalently, it assigns $a_{uv}$ to each vertex of $L(G)$. By Prop.~\ref{Aprop:inc_rob_bound},
\begin{equation}\label{eq:eta_line}
\eta(L(G))\;\le\;|E|^{-1}\,\max_{a:E\rightarrow \{\pm 1\}}\,\|H(a)\|_\infty,
\end{equation} 
where $H(a)=i\sum_{\{u,v\}\in E} a_{uv}\,\Gamma_u\Gamma_v$. We can rewrite
\begin{equation}\label{eq:K-form}
H(a)=\frac{i}{2}\sum_{\{u,v\}\in E}\!\bigl(a_{uv}\,\Gamma_u\Gamma_v - a_{uv}\,\Gamma_v\Gamma_u\bigr)
=\frac{i}{2}\sum_{u,v\in V} S_{uv}(a)\,\Gamma_u\Gamma_v
=\frac{i}{2}\,\Gamma^\top S(a)\,\Gamma,
\end{equation}
where $\Gamma^\top=(\Gamma_u)_{u\in V}$ and
\begin{equation}\label{eq:S-a-def}
S_{uv}(a)=
\begin{cases}
\;\;a_{uv}, & u<v \text{ and } \{u,v\}\in E,\\[2pt]
-a_{uv}, & v<u \text{ and } \{u,v\}\in E,\\[2pt]
\;\;0, & \text{otherwise}.
\end{cases}
\end{equation}
Thus $S(a)$ is skew--symmetric, with $\pm1$ in the off--diagonal positions corresponding to edges and zeros elsewhere. In particular, $S(a)$ is the skew--adjacency matrix of an orientation of $G$. 

Since $S(a)$ is real skew-symmetric, there exists an orthogonal matrix
$O\in\mathbb{R}^{n\times n}$ such that
\begin{equation}\label{eq:rotated_majoranas}
O^\top S(a)\,O
\;=\;
\bigoplus_{j=1}^{r}
\begin{pmatrix}
0 & \epsilon_j\\[2pt]
-\epsilon_j & 0
\end{pmatrix}
\end{equation}
and, if $n$ is odd, an additional $1\times1$ zero block, where
$r=\lfloor n/2\rfloor$ and $\epsilon_j\ge 0$. The eigenvalues of $S(a)$ are then $\{\pm i\epsilon_1,\ldots,\pm i\epsilon_r\}$ together with a single $0$ eigenvalue if $n$ is odd. Define the rotated Majorana operators by $\Gamma':=O^\top\Gamma$, such that
$\Gamma'_k = \sum_{v} O_{vk}\Gamma_v$ and the family $\{\Gamma'_k\}_{k\in V}$
still satisfies the Majorana relations~\eqref{eq:majoranas}. Substituting into
\eqref{eq:K-form} and applying \eqref{eq:rotated_majoranas} gives
\begin{equation}
H(a)\;=\;i\sum_{j=1}^{r}\epsilon_j\,\Gamma'_{2j-1}\Gamma'_{2j}\,.
\end{equation}
The operators $\{\Gamma'_{2j-1}\Gamma'_{2j}\}_{j=1}^r$ pairwise commute and each
squares to the identity, so they can be diagonalised simultaneously. The spectrum
of $H(a)$ is therefore $\{\sum_{j=1}^{r} s_j\,\epsilon_j \,|\, s_j\in\{\pm1\}\}$. Therefore, the operator norm is $
\|H(a)\|_\infty=\sum_{j=1}^r\epsilon_j$. Labelling the eigenvalues of $S(a)$ by $\lambda_1,\ldots,\lambda_n$, we obtain
\begin{equation}
\|H(a)\|_\infty
= \frac{1}{2}\sum_{j=1}^{n}|\lambda_j|\,.
\end{equation}

Let $\sigma$ be the orientation of $G$ for which $G^\sigma$ has skew-adjacency matrix $S(a)$ in \eqref{eq:K-form}. Its skew-energy (cf. Def. \ref{def:energies}) is therefore $\mathcal{E}_s(G^\sigma)=\sum_{j}|\lambda_j|$. It follows that $
\|H(a)\|_\infty=\mathcal{E}_s(G^\sigma)/2$. Maximising over all signings $a$ is equivalent to maximising over all orientations $\sigma$, therefore \eqref{eq:eta_line} becomes $\eta(L(G))\leq \tfrac{1}{2}|E|^{-1}\max_{\sigma}{E}_s(G^\sigma)$, proving the claim.

If $G$ is edge-transitive, Lemma~\ref{lem:uniform_rigid} implies that the associated quadratic Majorana measurement assemblage is uniform and rigid. As shown in~\cite{nguyen20}, the upper bound in Prop.~\ref{Aprop:inc_rob_bound} is tight for uniform and rigidly symmetric measurement assemblages. Since Prop.~\ref{Aprop:inc_rob_bound} reduces to~\eqref{Aeq:skew_bound_eigs} in the line-graph setting, we conclude that equality holds when $G$ is edge-transitive.
\end{proof}

\section{Incompatibility robustness for families of line graphs}\label{A:line_examples}

\subsection*{Cycles}

\begin{corollary*}[Restatement of Cor. \ref{prop:eta-cycle}]\label{Aprop:eta-cycle}
For cycle graphs $C_n$, incompatibility robustness is given by
 \begin{equation}\label{Aeq:eta-cycle}
        \eta(C_n)= \begin{cases}
 \frac{1}{n}\,\cot(\pi/2n), & n \text{ odd},\\
\frac{2}{n}\,\csc(\pi/n), & n \text{ even.}
\end{cases}
    \end{equation}
As $n\rightarrow\infty$, $\eta(C_n)\rightarrow\frac{2}{\pi}$.
\end{corollary*}

\begin{proof}
$C_n$ is edge-transitive, therefore Thm.~\ref{thm:spectra} implies
$\eta(C_n)
=\frac{1}{2n}\mathcal{E}_s^{\max}(C_n)$.
A cycle has exactly two switching classes (determined by the sign of the unique cycle), and their skew-energies are given in \cite{adiga10}. For $n$ odd, both classes have $
\mathcal{E}_s(C_n^\sigma)=2\,\cot(\frac{\pi}{2n})$, with $\sigma$ the orientation. For $n$ even, the two values are $
\mathcal{E}_s(C_n^\sigma)=4\,\cot(\frac{\pi}{n})$ and $
\mathcal{E}_s(C_n^{\sigma'})=4\,\csc(\frac{\pi}{n})$. Taking the maximum gives the desired formula \eqref{eq:eta-cycle}. Using $\cot x=\tfrac1x-\tfrac{x}{3}+O(x^{3})$ and $\csc x=\tfrac1x+\tfrac{x}{6}+O(x^{3})$ with $x=\tfrac{\pi}{2n}$ (odd) and $x=\tfrac{\pi}{n}$ (even), we obtain $\eta(C_n)=\tfrac{2}{\pi}-\tfrac{\pi}{6n^{2}}+O(n^{-4})$ (odd) and $\eta(C_n)=\tfrac{2}{\pi}+\tfrac{\pi}{3n^{2}}+O(n^{-4})$ (even), hence $\eta(C_n)\to 2/\pi$.
\end{proof}

\subsection*{Johnson graphs}

\begin{corollary*}[Restatement of Cor. \ref{cor:johnson}]\label{prop:eta-triangular}
For Johnson graphs $J(n,2)$, the incompatibility robustness of the corresponding set of $\binom{n}{2}$ observables satisfies
\begin{equation}\label{Aeq:eta-triangular}
\eta\big(J(n,2)\big) \;\le\; \frac{1}{\sqrt{\,n-1\,}}\,,
\end{equation}
with equality if and only if a skew--conference matrix of order $n$ exists.
\end{corollary*}

\begin{proof}
The skew-adjacency matrix of any oriented complete graph $K_n^\sigma$ is a \emph{tournament matrix}: a real
$n\times n$ skew-symmetric matrix $T=(t_{ij})$ with zero diagonal and
$t_{ij}\in\{\pm1\}$ for $i\neq j$. Let $\mathcal{T}_n$ denote the set of all such matrices. Since $K_n$ is edge-transitive, Thm.~\ref{thm:spectra} implies
\[
\eta\big(J(n,2)\big)
\;=\; \frac{1}{2|E(K_n)|}\,\mathcal{E}_s^{\max}(K_n)
\;=\; \frac{1}{n(n-1)}\,\max_{T\in \mathcal{T}_n}\sum_{j=1}^{n}\big|\lambda_j(T)\big|\,,
\]
where $\lambda_j(T)$ are the eigenvalues of $T$. As shown in \cite[Prop. 4.1]{ito17}, any tournament matrix $T\in\mathcal{T}_n$ satisfies
$\sum_j |\lambda_j(T)| \;\le\; n\sqrt{n-1}$,
with equality if and only if $T$ is a skew-conference matrix. This yields the required bound.
\end{proof}

\subsection*{Line graphs of hypercubes}

\begin{corollary*}[Restatement of Cor. \ref{cor:eta-hypercube-exact}]\label{Acor:eta-hypercube-exact}
Let $d\geq 3$. For the line graph $L(Q_d)$ of the $d$-dimensional hypercube $Q_d$, the incompatibility robustness of the corresponding set of $d2^{d-1}$ observables satisfies
\begin{equation}
\eta\big(L(Q_d)\big)\;=\;\frac{1}{\sqrt{d}}.
\end{equation}
\end{corollary*}

\begin{proof}
Since $Q_d$ is edge-transitive, Thm.~\ref{thm:spectra} yields $\eta\big(L(Q_d)\big)=\frac{1}{d2^d}\mathcal{E}_s^{\max}(Q_d)$, with
$|E(Q_d)|=d\,2^{d-1}$. A general bound for the maximal skew-energy of any graph is $\mathcal{E}_s^{\max}(G)\le n\sqrt{\Delta}$,
where $n=|V(G)|$ and $\Delta$ is the maximum degree \cite{adiga10}. For $Q_d$ we have $n=2^d$ and
$\Delta=d$. It is known that, for every $d\geq3 $, there exists an orientation $\sigma$ such that $Q_d^\sigma$ saturates this bound \cite{tian11}, i.e., $\mathcal{E}_s^{\max}(Q_d)=n\sqrt{d}=2^d\sqrt{d}$. Therefore, $\eta\big(L(Q_d)\big)
=\frac{1}{\sqrt{d}}$.
\end{proof}

\subsection*{Rook's graphs}

\begin{corollary*}[Restatement of Cor. \ref{cor:eta-LKmn}]\label{Acor:eta-LKmn}
Let $s\geq r \geq 1$. For the rook's graphs $K_r\square K_s$, the incompatibility robustness of the corresponding set of $rs$ observables satisfies
\begin{equation}
\eta\big(K_r\square K_s\big)\;\le\;\frac{1}{\sqrt{s}},
\end{equation}
with equality if and only if a partial Hadamard matrix of size $r\times s$ exists.
\end{corollary*}

\begin{proof}
Let $K_{r,s}$ be the complete bipartite graph with $r+s$ vertices and $rs$ edges, such that $L(K_{r,s})=K_r\square K_s$.  
For any orientation $K_{r,s}^{\tau}$, bipartiteness implies that its $(r+s)\times (r+s)$ skew-adjacency 
matrix has the block structure
\begin{equation}
S=\begin{pmatrix}
0 & C\\[2pt]
-\,C^\top & 0
\end{pmatrix},
\end{equation}
where the matrix $C\in\{\pm 1\}^{r\times s}$ encodes the signs determined by the orientation~$\tau$. Let $\{\sigma_j(C)\}_{j=1}^{\ell}$ denote the singular values of $C$,
with $\ell=\operatorname{rank}C\le \min\{r,s\}=r$.
Because $S$ is skew-symmetric, its eigenvalues are
$\{\pm i\,\sigma_j(C)\}_{j=1}^{\ell}\cup\{0,\ldots,0\}$.
Hence the skew-energy is $\mathcal{E}_s(K_{r,s}^{\tau})=2\sum_{j=1}^{\ell}\sigma_j(C)$,
and therefore, by Theorem~\ref{thm:spectra},
\begin{equation}\label{eq:rook_singular_rs}
\eta(L(K_{r,s}))
=\frac{1}{2|E(K_{r,s})|}\,\mathcal{E}_s^{\max}(K_{r,s})
=\frac{1}{rs}\,\max_{C\in\{\pm1\}^{r\times s}}\sum_{j=1}^{\ell}\sigma_j(C).
\end{equation}
By Cauchy-Schwarz and the Frobenius norm,
\begin{equation}\label{eq:cauchy_rs}
\sum_{j=1}^{\ell}\sigma_j(C)
\leq \sqrt{\ell}\sqrt{\sum_{j=1}^{\ell}\sigma_j(C)^2}
= \sqrt{\ell}\,\|C\|_F
\le \sqrt{r}\,\sqrt{rs}
= r\sqrt{s},
\end{equation}
where we used $\|C\|_F^2=\sum_{i,j}C_{ij}^2=rs$. Combining \eqref{eq:rook_singular_rs} and \eqref{eq:cauchy_rs} yields
\begin{equation}
\eta(L(K_{r,s}))\le\frac{1}{rs}\,r\sqrt{s}=\frac{1}{\sqrt{s}}.
\end{equation}
The inequalities in \eqref{eq:cauchy_rs} are tight if and only if 
$\ell = r$ and all nonzero singular values of $C$ are equal, which is equivalent to $\sigma_1(C) = \cdots = \sigma_r(C) = \sqrt{s}$. Since the eigenvalues of $CC^\top$ are $\{\sigma_j(C)^2\}_{j=1}^r$, 
this condition holds if and only if $CC^\top = s\,\id_r$. 
Thus the inequalities in \eqref{eq:cauchy_rs} are tight if and only if 
$C$ is an $r\times s$ partial Hadamard matrix.
\end{proof}

\subsection*{Paths}\label{A:path}

In contrast to odd cycles (cf. Fig. \ref{fig:eta-side-by-side}), paths do not lose incompatibility as the number of vertices increase, therefore
\begin{equation}\label{eq:paths}
    \eta(P_n)\;\geq \;\eta(P_{n+1})\,,
\end{equation}
for all $n\geq 1$. This follows from the observation that $P_n$ can be obtained from $P_{n+1}$ by removing a vertex: a subset of observables cannot be more incompatible than the original set. Cycles, on the other hand, become paths when a vertex is removed, therefore, by repeated application of \eqref{eq:paths} we have, for all $m<n$,
\begin{equation}\label{eq:cycles}
    \eta(P_m)\;\geq \;\eta(C_{n})\,.
\end{equation}

We use this observation, together with Thm. \ref{thm:spectra} and the incompatibility formula for cycles in Cor. \ref{prop:eta-cycle}, to prove the following result.

\begin{corollary*}[Restatement of Cor. \ref{prop:eta_paths}]
    \label{Aprop:eta_paths}
    Let $P_n$ be a path on $n$ vertices. For $n$ even,
\begin{equation}\label{eq:path_even}
 \frac{2}{n+2}\,\csc\!\Big(\frac{\pi}{n+2}\Big) \leq \eta(P_n)\leq \frac{1}{n}\,\cot\!\Big(\frac{\pi}{2(n+2)}\Big) - \frac{1}{n}\,,
    \end{equation}
and for $n$ odd,
  \begin{equation}\label{eq:path_odd}
 \frac{2}{n+1}\,\csc\!\Big(\frac{\pi}{n+1}\Big) \leq \eta(P_n)\leq \frac{1}{n}\,\csc\!\Big(\frac{\pi}{2(n+2)}\Big) - \frac{1}{n}\,.
    \end{equation}
Furthermore, $\eta(P_n)\rightarrow\frac{2}{\pi}$ as $n\rightarrow\infty$.
\end{corollary*}

\begin{proof}
Paths are bipartite, so there exists an orientation $\sigma$ such that $\mathcal{E}_s(P_n^{\sigma})=\mathcal{E}(P_n)$ \cite{shader09}. Since paths contain no cycles, all orientations belong to the same switching class (see \cite[Thm. 3.3]{adiga10}), therefore the skew-energy does not depend on the orientation and $\mathcal{E}_s^{\max}(P_n)=\mathcal{E}(P_n)$. Applying Thm. \ref{thm:spectra} to $P_n=L(P_{n+1})$, we obtain $\eta(P_n)\leq \mathcal{E}(P_{n+1})/2n$. The exact formula for $\mathcal{E}(P_{n+1})$ can be found, for example, in \cite{gutman12} and gives the upper bounds in \eqref{eq:path_even} and \eqref{eq:path_odd}. The lower bounds follow from \eqref{eq:cycles} and the formula for the incompatibility robustness of cycles, given in \eqref{eq:eta-cycle}. We conclude from Cor. \ref{prop:eta-cycle} that paths exhibit the same asymptotic behaviour as cycles.
\end{proof}

\section{Merged Johnson graphs}\label{A:degreeq}

As we have seen, the anti-commutation relations of the quadratic observables $\{i\Gamma_u\Gamma_v\,|\,\{u,v\}\in \binom{[n]}{2}\}$ are described by the Johnson graph $J(n,2)$. Generalised Johnson graphs are the building blocks that describe the anti-commutativity graphs of degree--$k$ Majorana observables $\mathcal{A}_k=\{A_{\mathscr{I}}\,|\,\mathscr{I}\in\binom{[n]}{k}\}$ defined in \eqref{eq:ferm_monomials}.
\begin{definition}\label{def:generalised_johnson}
The \emph{generalised Johnson graph} $J(n,k,i)$ has vertex set $V=\binom{[n]}{k}$, with $\mathscr{I},\mathscr{I}'\in V$ adjacent if and only if $|\mathscr{I}\cap \mathscr{I}'|=i$.
\end{definition}
The Johnson graph $J(n,2)$ is equivalent to $J(n,2,1)$, and is illustrated in Fig. \ref{fig:line} for $n=5$. By considering a graph whose edge set consists of the union of generalised Johnson graphs, we obtain a \emph{merged} Johnson graph \cite{jones05}: 
\begin{definition}\label{def:union_johnson}
The \emph{merged Johnson graph} $J(n,k,L)$ with $L\subseteq \{0,1,\ldots, k\}$ has vertex set $V=\binom{[n]}{k}$, with $\mathscr{I},\mathscr{I}'\in V$ adjacent if and only if $|\mathscr{I}\cap \mathscr{I}'|\in L$.\end{definition}
In particular, the edge set of $J(n,k,L)$ is given by the union of the edge sets of the generalised Johnson graphs $J(n,k,i)$ with $i\in L$. 

The $L$-set necessary to describe the anti-commutativity graph of $\mathcal{A}_k$ depends on whether $k$ is odd or even. For odd $k$, the relevant set is $L_{\text{even}}\equiv \{0,2,4,\ldots,k-1\}$, while for even $k$, the set is $L_{\text{odd}}\equiv \{1,3,5,\ldots,k-1\}$. In particular, the anti-commutativity graph of $\mathcal{A}_k$ is
\begin{equation}\label{eq:anticommutativity_johnson}
G=\begin{cases}
J(n,k,L_{\text{even}}) &\mbox{if} \,\,\, k\,\, \mbox{is odd} \\
J(n,k,L_{\text{odd}}) &\mbox{otherwise} \,.
 \end{cases}
\end{equation}

While a full characterisation of the relevant graph parameters of $G$ in \eqref{eq:anticommutativity_johnson} is unknown, some asymptotic bounds can be derived. To calculate a lower bound on $\eta(G)$ we will make use of a result by Jones \cite{jones05} on the automorphism groups of merged Johnson graphs.

\begin{theorem}[Jones~{\cite[Thm.~2]{jones05}}]\label{thm:jones}
For $2\leq k\leq n/2$ and $L\subset \{1,2,\ldots,k\}$,  the automorphism group of $J(n,k,L)$ (excluding $J(12,4,\{1,3\})$ and $J(12,4,\{2,4\})$) is the symmetric group $S_n$.
\end{theorem}

This allows us to evaluate the fractional chromatic number of $J(n,k,L_{\text{odd}})$.

\begin{corollary}\label{cor:johnson_fractional}
Suppose $n$ and $k$ are even, and $2\leq k\leq n/2$. Then, for all sufficiently large $n$, the fractional chromatic number of the Johnson graph 
$J(n,k,L_{\mathrm{odd}})$ satisfies
\begin{equation}\label{eq:johnson_chromatic}
\chi_f(J(n,k,L_{\text{odd}}))=
\binom{n}{k}\binom{n/2}{k/2}^{-1}\,.
\end{equation}
\end{corollary}
\begin{proof} 
For any pair $\mathscr{I},\mathscr{I}'\in\binom{[n]}{k}$, there exists $\pi\in S_n$ such that $\pi(\mathscr{I})=\mathscr{I}'$. Therefore, by Thm. \ref{thm:jones}, the merged Johnson graphs $J(n,k,L)$, under the given assumptions, are vertex-transitive (see Def. \ref{def:transitive}). Consequently, the fractional chromatic number is given by $\chi_f(G)=|V|/\alpha(G)$, where $\alpha(G)$ is the independence number \cite{scheinerman11}. The relevant independence numbers have been partially characterised in work by Deza, Erd\"os, and Frankl \cite{deza78} (see also \citep[Thm.~3.11]{hastings21}). In particular, if $k$ and $n$ are even, with $n$ sufficiently large, then $\alpha(J(n,k,L_{\text{odd}}))=\binom{n/2}{k/2}$.
\end{proof}

To bound $\eta(G)$ from above, we consider either the Lovász number or the clique number of merged Johnson graphs. While a general bound on the clique number is unknown, we can derive (see Appendix \ref{A:lsystems}) asymptotic bounds for specific $k$ using results from extremal set theory \cite{frankl16}. However, stronger bounds can be derived from the following result.

\begin{theorem}[Linz \cite{linz24}]\label{thm:hastings_lovasz}
Let $k=\mathcal{O}(1)$. For the merged Johnson graph $J(n,k,L)$, the Lovász number satisfies
\begin{equation}
\vartheta(J(n,k,L))=\Theta(n^{k-|L|})\,.
\end{equation}
\end{theorem}
This extends a bound in \cite{hastings21} that showed, for even $k\leq 10$ and $n$ sufficiently large, that $\vartheta(J(n,k,L_{\text{odd}}))\leq  \binom{n/2}{k/2}$.

\section{$L$-systems}\label{A:lsystems}

To find asymptotic bounds on the clique number of merged Johnson graphs $J(n,k,L)$, we consider some results from extremal set theory, particularly $L$-systems (see, e.g., \emph{Invitation to intersection problems for finite sets} by Frankl and Tokushige \cite{frankl16}).

\begin{definition}
Fixing $n$ and $k$, let $L\subset \{0,1,\ldots,k-1\}$. A $(n,k,L)$-system is a family $\mathcal{F}\subset\binom{[n]}{k}$ of $k$--element subsets such that any distinct $F,F'\in\mathcal{F}$ satisfies $|F\cap F'|\in L$.
\end{definition}

A well-studied problem is to find, for a given $n,k$ and $L$, the maximal cardinality $m(n,k,L)$ of all possible $(n,k,L)$-systems. In particular, the task is to find the so-called \emph{exponent} $\tau$ of the $(n,k,L)$-system such that $c\cdot n^{\tau}< m(n,k,L)< c'\cdot n^{\tau}$. In general, evaluating $\tau$ is an open problem, but a full characterisation is known for $k\leq 12$ \cite{frankl96}. The cases relevant in this work are summarised in the following theorem. 

\begin{theorem}[Frankl, Ota, and Tokushige \cite{frankl96}]\label{thm:frankl}
Suppose $k\leq 12$ and let $\tau$ denote the exponent of a $(n,k,L)$-system. If $k$ is even and $L= \{1,3,\ldots k-1\}$ then $\tau=1$. If $k$ is odd and $L=\{0,2,\ldots,k-1\}$ then $\tau=2$.
\end{theorem}

If $k$ is even and $L=\{1,3,\ldots,k-1\}$, a $(n,k,L)$-system corresponds to a clique of $J(n,k,L_{\text{odd}})$. Therefore, the largest $(n,k,L)$-system is equivalent to the maximum clique, i.e., $m(n,k,L)=\omega(J(n,k,L_{\text{odd}}))$. Similarly, if $k$ is odd and $L=\{0,2,\ldots,k-1\}$, then $m(n,k,L)=\omega(J(n,k,L_{\text{even}}))$. This observation, together with Thm. \ref{thm:frankl}, leads directly to the following asymptotic bound on the clique number.

\begin{corollary}\label{prop:clique_johnson}
Let $k\leq 12$, and $n\geq 2k$. For $k$ even, $\omega(J(n,k,L_{\text{odd}}))=\Omega(n^2)$. For $k$ odd, $\omega(J(n,k,L_{\text{even}}))=\Omega(n)$.
\end{corollary}

Cor. \ref{prop:clique_johnson}, together with Prop. \ref{prop:clique}, provide a sub-optimal upper bound on the incompatibility robustness for $\mathcal{A}_k$. A stronger bound can be found by considering the Lovász number of $J(n,k,L)$, as shown in Appendix \ref{A:degreeq}.

\end{document}